\documentclass{article}

\usepackage[T1]{fontenc}
\usepackage[utf8]{inputenc}
\usepackage[english]{babel}
\usepackage{tipa}
\usepackage[usenames,dvipsnames]{xcolor}
\usepackage{amsthm, mathtools}
\usepackage{amsmath}
\usepackage{amssymb}
\usepackage{graphicx}
\usepackage{mathrsfs}
\usepackage{algorithm}
\usepackage{xspace}
\usepackage{thmtools}
\usepackage{enumitem}
\usepackage{complexity}
\usepackage{tikz}
\usepackage{circuitikz}

\usepackage{url}

\usepackage{amsmath}
\usepackage{graphicx}
\usepackage[colorlinks=true, allcolors=blue]{hyperref}
\usepackage{cleveref}
\usepackage{verbatim}
\usepackage{longtable}
\usepackage{array}

\creflabelformat{equation}{#2\textup{#1}#3}

\theoremstyle{definition}
\declaretheorem[style=definition,qed=$\lrcorner$]{theorem}
\declaretheorem[style=definition,qed=$\lrcorner$,sibling=theorem]{corollary}
\declaretheorem[style=definition,qed=$\lrcorner$]{definition}
\declaretheorem[style=definition,qed=$\lrcorner$,sibling=theorem]{lemma}
\declaretheorem[style=definition,qed=$\lrcorner$,sibling=theorem]{conjecture}
\declaretheorem[style=definition,qed=$\lrcorner$,sibling=theorem]{claim}
\declaretheorem[style=definition,qed=$\lrcorner$]{example}
\crefname{claim}{claim}{Claim}
\declaretheorem[style=definition,qed=$\lrcorner$]{remark}

\newcommand{\bigO}{\ensuremath{\mathcal{O}}\xspace}
\newcommand{\ER}{\ensuremath{\exists\mathbb{R}}\xspace}
\renewcommand{\P}{\ensuremath{\textsc{P}}\xspace}
\renewcommand{\NP}{\ensuremath{\textsc{NP}}\xspace}
\renewcommand{\R}{\ensuremath{\mathbb{R}}\xspace}

\newcommand{\N}{\ensuremath{\mathbb{N}}\xspace}
\newcommand{\Z}{\ensuremath{\mathbb{Z}}\xspace}
\newcommand{\ETR}{\textsc{ETR}\xspace}
\newcommand{\multiplier}{\ensuremath{4860}\xspace}
\renewcommand{\AC}{\ensuremath{\textsc{AC}^0}\xspace}
\newcommand{\TrulyRPH}{\ensuremath{\mathsf{T}\mathbb{R}\textsc{PH}}\xspace}
\newcommand{\RPH}{\ensuremath{\mathbb{R}\textsc{PH}}\xspace}
\newcommand{\LogicRPH}{\ensuremath{\texttt{Logic-}\mathbb{R}\textsc{PH}}\xspace}

\renewcommand{\PH}{\textsc{PH}\xspace}

\renewcommand{\PR}{\ensuremath{\textsc{P\R}}\xspace} 

\renewcommand{\BQP}{\textsc{BQP}\xspace}
\renewcommand{\PSPACE}{\textsc{PSPACE}\xspace}

\newcommand{\oracle}[1]{#1}

\newcommand{\Size}[1]{\mathbf{SIZE}(#1)}

\newcommand{\LLL}{\ensuremath{\mathcal{L}}}
\newcommand{\algo}[1]{#1}
\newcommand{\quasipoly}[1]{\mathtt{quasipoly}(#1)}
\newcommand{\sipser}{\ensuremath{F^{k+1}_N}}

\DeclareMathOperator*{\probability}{\ensuremath{\textsf{Pr}}}

\newcommand{\PosSLP}{\ensuremath{\textsc{PosSLP}}\xspace}
\newcommand{\bin}{\ensuremath{\texttt{bin}}\xspace}
\newcommand{\real}{\ensuremath{\texttt{real}}\xspace}
\newcommand{\RPHcircuitUpperBound}{S}

\newcommand{\truthtable}{\ensuremath{\textsf{tt}}}

\renewcommand{\poly}{\ensuremath{\texttt{poly}}}

\renewcommand{\BP}{\mathsf{BP}}

\newcommand{\DRPHlevel}[1]{\ensuremath{\Delta_{#1}\mathbb{R}}}
\newcommand{\SRPHlevel}[1]{\ensuremath{\Sigma_{#1}\mathbb{R}}}
\newcommand{\PRPHlevel}[1]{\ensuremath{\Pi_{#1}\mathbb{R}}}

\newcommand{\TrulyDRPHlevel}[1]{\ensuremath{\Delta_{#1}\mathsf{T}\mathbb{R}}}
\newcommand{\TrulySRPHlevel}[1]{\ensuremath{\Sigma_{#1}\mathsf{T}\mathbb{R}}}
\newcommand{\TrulyPRPHlevel}[1]{\ensuremath{\Pi_{#1}\mathsf{T}\mathbb{R}}}

\newcommand{\DPHlevel}[1]{\Delta_{#1}}
\newcommand{\SPHlevel}[1]{\Sigma_{#1}}
\newcommand{\PPHlevel}[1]{\Pi_{#1}}

\newcommand{\Parity}{\textsc{Parity}\xspace}
\newcommand{\ParityA}{\ensuremath{\textsc{Parity}_O}\xspace}
\newcommand{\Parityn}{\ensuremath{\textsc{Parity}}\xspace}

\newcommand{\operations}{\texttt{op}}
\newcommand{\signset}{\texttt{sign}}
\newcommand{\power}{\texttt{power}}

\newcommand{\CNF}{\textsc{CNF}}
\newcommand{\DNF}{\textsc{DNF}}
\newcommand{\leaf}{\texttt{leaf}}
\newcommand{\route}{\texttt{path}}

\newcommand{\T}{T}
\newcommand{\prop}{\ensuremath{\texttt{Prop}}\xspace}

\usepackage[letterpaper,top=2cm,bottom=2cm,left=3cm,right=3cm,marginparwidth=1.75cm]{geometry}

\newcommand{\lucas}[1]{
}
\newcommand{\till}[1]{
\textcolor{orange}{Till: #1}
}
\newcommand{\subha}[1]{
}
\newcommand{\thekla}[1]{
}
\newcommand{\reviewer}[1]{
}
\newcommand{\new}[1]{\textcolor{blue}{#1}}

\title{Oracle Separations for~$\RPH$}

\author{Thekla Hamm\thanks{Eindhoven University of Technology. {\tt t.l.s.hamm@tue.nl}}
\and
Lucas Meijer\thanks{Utrecht University. {\tt l.meijer2@uu.nl} is generously supported by the Netherlands Organisation for Scientific Research (NWO) under project no. VI.Vidi.213.150.}
\and 
Tillmann Miltzow\thanks{Utrecht University. {\tt t.miltzow@uu.nl} is generously supported by the Netherlands Organisation for Scientific Research (NWO) under project no. 016.Veni.192.250 and no. VI.Vidi.213.150.}
\and 
Subhasree Patro\thanks{Eindhoven University of Technology. {\tt s.patro@tue.nl}}}
\date{}

\begin{document}

\maketitle

\begin{abstract}
    While theoretical computer science primarily works with discrete models of computation, like the Turing machine and the wordRAM, there are many scenarios in which introducing real computation models is more adequate.
    For example, when working with continuous probability distributions for say smoothed analysis, in continuous optimization, computational geometry or machine learning.
    We want to compare real models of computation with discrete models of computation. 
    We do this by means of oracle separation results.
    
    We define the notion of a \textit{real Turing machine} as an extension of the (binary) Turing machine
    by adding a real tape.
    Using those machines, we define and study the  real polynomial hierarchy \RPH.
    We are interested in \RPH as the first level of the hierarchy corresponds to the well-known complexity class \ER. 
    It is known that $\NP \subseteq \ER \subseteq \PSPACE$
    and furthermore $\PH \subseteq \RPH \subseteq \PSPACE$.
    We are interested to know if any of those inclusions are tight.
    In the absence of unconditional separations of complexity classes, we turn to oracle separation.
    We develop a technique that allows us to transform oracle separation results from the binary world to the real world. 
    As applications, 
    we show there are oracles such that:
    \begin{itemize}
        \item $\RPH^O$ proper subset of $\PSPACE^O$,
        \item $\SPHlevel{k+1}^O$ not contained in $\SRPHlevel{k}^O$, for all $k\geq 0$,
        \item $\SRPHlevel{k}^O$ proper subset of $\SRPHlevel{k+1}^O$, for all $k\geq 0$,
        \item $\BQP^O$ not contained in $\RPH^O$.
    \end{itemize}
    Our results hint that \ER is strictly contained in \PSPACE and that there is a separation between the different levels of the real polynomial hierarchy. We also bound the power of real computations by showing that \NP-hard problems are unlikely to be solvable using polynomial time on a realRAM. 
    Furthermore, our oracle separations hint that polynomial-time quantum computing cannot be simulated on an efficient real Turing machine.
\end{abstract}

\vfill

\newpage

\tableofcontents

\newpage
\section{Introduction}


\subsection{Motivation}
    In theoretical computer science, discrete models of computation are dominant.
    They capture many algorithmic problems accurately and have a rich structural theory~\cite{AB09}.
    Yet, in many contexts, continuous numbers are more appropriate or even necessary to be used as part of the model of computation.
    The primary examples are computational geometry~\cite{2017handbookDCG}, machine learning~\cite{alpaydin2020introduction}, continuous optimization~\cite{2020introductionContinuousOptimiziation} and working with probabilities~\cite{spielman2004smoothed}.
    Yet, working with continuous numbers and thus inherently with infinite precision makes many researchers uneasy.
    The reasons for that are multitude.
    Firstly, it is clearly an unrealistic model of computation, as a physical computer can only work with finite precision.
    Secondly, it might be possible to exploit the real number computations to say solve problems much faster
    using real numbers compared to using binary numbers.
    For example, one such exploitation might be the ability to solve \PSPACE-complete problems in polynomial time using real numbers.
    For instance, Shamir~\cite{shamir1977factoring} shows how to factor integers,
    and Sch\"{o}nhage showed how to solve general \PSPACE-complete problems~\cite{schonhage1979power}.
    However, both algorithms use the rounding operation in a clever way.
    Since then, the realRAM is considered not to have the rounding operation~\cite{EvdHM20}.
    And without rounding, we are not aware of any other anomalies or exploitations with real number computations. 
    Still, the fact that researchers have not found any new abnormalities does not mean that they do not exist.
    At last, we want to point out that real Turing machines are often considered with access to arbitrary real constants. 
    In this case, it was shown that those models of computation are at least as powerful as \P/\poly~\cite{Koiran94} and the best upper bound is 
    \PSPACE/\poly~\cite{CuckerGrigoriev1997}.
       
    We want to show \textit{limitations} 
    on real number computations.
    Most known methods to separate complexity classes unconditionally use the diagonalization method~\cite{AB09}. 
    As the diagonalization method is limited in scope, 
    we turn our attention to oracle separations.

\subsection{The Setup}
    \paragraph{Model of Real Computation.}
    There are several ways to define computations with real numbers.
    The most common way to work with real numbers is the realRAM model
    of computation~\cite{EvdHM20}.
    Its advantage is that it closely resembles modern computers and 
    thus allows fine-grained computational analysis.
    In other words, we can distinguish between linear, quadratic time computations etc.
    However, that model of computation is very involved, specifically as it allows random access.
    (This means that the content of one word register is used to address the content of another register.)

    The second common model is the Blum Shub Smale model (or BSS model) of computation, which allows for real inputs 
    and real outputs~\cite{BSS89, BCSS98}. 
    Its advantage is that it is closer to a Turing machine and thus more minimal.
    The downside is that it does not closely correspond to modern computer architectures.
    The research line building upon the BSS model often focuses on real input, where we focus on binary inputs $\{0,1\}^*$.
    Furthermore, originally the BSS model allowed the machine to have access to arbitrary real constants, whereas we are constant free.
    
    As there are many definitions of BSS machines, compare~\cite{BSS89} and \cite{BCSS98} for example, and as we want to work with real and binary input, we give our own definition of a real Turing machine.
    For us a real Turing machine has two tapes, one binary tape and one real tape. 
    Furthermore, we are allowed the usual arithmetic operations: addition, subtraction, multiplication, and division.
    See \Cref{sec:TuringMachines} for a detailed definition.
    The real Turing machine we define is convenient for our purpose, and we show that it is equivalent to the previous models up to polynomial factors. 
    Therefore, readers familiar with any other standard model of real computation may skip \Cref{sec:TuringMachines}.
    
    \paragraph{Real Polynomial Hierarchy.}
        We say that a language $L \subseteq \{0,1\}^*$ is in the real polynomial hierarchy $\RPH$ 
        if there is a constant $k\geq 1$, a univariate polynomial $q$, a polynomial-time real Turing machine $M$, and a sequence of (alternating) quantifiers $Q_1, \dotsc, Q_k \in \{\exists, \forall\}$ such that
        \begin{equation*}
        x \in L \iff Q_1 u_1 \in \R^{q(|x|)} Q_2 u_2 \in \R^{q(|x|)} \cdots Q_k u_k \in \R^{q(|x|)} M(x, u_1,\ldots,u_k)=1.
        \end{equation*}
        The $k$th level of the hierarchy~\RPH consists of all languages decidable in this way using \(k\) quantifiers.
        Languages on the \(k\)-th level are distinguished based on their first quantifier into \SRPHlevel{k} (existential), and \PRPHlevel{k} (universal).
    
        One can also define the \emph{truly} real polynomial hierarchy~\TrulyRPH in which even languages (rather than just computations) are over real numbers and not only over Boolean values.
        The truly real polynomial hierarchy is interesting in its own right, but not the topic of this article.
        It was considered in Chapter 21 by Blum, Cucker, Shub and Smale in the context of machines that are not allowed to do multiplication~\cite[Chapter 21]{BCSS98}.
        For general real Turing machines, we only found a paper by B\"{u}rgisser and Cucker~\cite{BC09} that defines this hierarchy, although the article was really about the hierarchy as we define it.
     
        To us, the complexity classes $\PR = \SRPHlevel{0}$ and $\ER = \SRPHlevel{1}$ are of primary importance. 
        The complexity class \PR and \ER are the real analogs of $\P$ and \NP. 
        The reason we find them most interesting is that 
        there are more problems studied for those two classes compared to other levels of the real polynomial hierarchy.
        See \Cref{sub:RealComputation} for some selected examples of problems in \PR and
        see \Cref{sec:ETR} for an introduction to the existential theory of the reals~\ER.
         
    \paragraph{Oracles.}
        When we consider an oracle model of computation, we assume that there is a set $O$ and our model of computation is allowed to make queries of the form $x\in O$ and receives the answer in one time step.    
        See \Cref{sec:OracleDefinition} for formal details, and \Cref{sec:discussion} for more context on why oracle separations are arguably considered a step towards separating complexity classes.
        We want to say already here, that oracle separation results serve as a \textit{hint} that complexity classes are separate.
        We discuss in \Cref{sub:OracleHistory} how to judge this ``hint''.

%
        
\subsection{Results}

    Our first result separates the real polynomial hierarchy 
    from \PSPACE using random oracles.

    \begin{restatable}{theorem}{PSPACEnotRPH}
        \label{thm:PSPACEnotRPH}
        For a random binary oracle $O$,
        $\PSPACE^O \nsubseteq \RPH^O$ with probability 1.
    \end{restatable}

    To understand the significance of the theorem, let us consider the complexity class $\ER = \SRPHlevel{1}$, which is on the first level of the hierarchy.
    It is known that $\NP \subseteq \ER \subseteq \PSPACE$.
    \Cref{thm:PSPACEnotRPH} implies that $\PSPACE^O \nsubseteq \ER^O$ and thus gives the first hint that \ER is not equal to \PSPACE.
    Again, looking at the gap between \NP and \PSPACE, we see that there are infinitely many layers of the polynomial hierarchy in between.

    As \ER-hard problems are arguably harder to solve than
    problems in \NP
    it is plausible that \ER 
    could for instance contain more levels of the polynomial hierarchy.
    The next theorem hints that this is not the case.
    
    \begin{restatable}{theorem}{RPHlayers}
    \label{thm:RPHlayers}
        For each $k\geq 0$,
        for a random oracle $O$, with probability $1$ $\SPHlevel{k+1}^{O} \nsubseteq \SRPHlevel{k}^O$. 
    \end{restatable}

    Note that the previous theorem implies the following corollary by choosing $k=0$.
    (We denote by $\PR$ polynomial-time computations on a real Turing machine with binary input, also denoted as $\text{BP}^0(P_{\R})$ in the literature.)

    \begin{corollary}
        \label{cor:NP-notin-PR}
        There exists a
        binary oracle $O$ such that $\NP^O \nsubseteq \PR^O$.    
    \end{corollary}
    This hints that polynomial time realRAM computation 
    is likely not capable of solving \NP-hard problems.
    (Specifically, this hints that the 
    result by B\"{u}rgisser and Jindal~\cite{BJ24} about NP-hardness of \PosSLP is not true in full generality, see explanation on \PosSLP below.)
    In other words, their assumptions are too strong.
    %
    It is trivial that $\SPHlevel{k} \subseteq \SRPHlevel{k}$, for every $k$, as we can restrict any real witness to be binary.
    This implies that we also have a separation of the  different levels within \RPH.

    \begin{corollary}
        For each $k\geq 0$, there exists a binary oracle $O$ such that $\SRPHlevel{k}^O \subsetneq \SRPHlevel{k+1}^O$.
    \end{corollary}

    At last, we turn our attention to quantum computing.
    We were wondering if \ER would be powerful enough to simulate quantum computing.
    Quantum computing relies on quantum mechanics, where qubits can exist in superposition, representing a continuous range of states. 
    Quantum algorithms manipulate these states through interference and entanglement, which allows for exploration of solution spaces differently than classical computers. 
    Although the number of amplitudes to describe a state in superposition can be exponential,  we hoped that one might be able to simulate those 
    quantum computations using real numbers. 
    
    The idea was that the amplitudes of all quantum states could be stored by real numbers that could even be guessed by a real Turing machine, i.e., $\BQP \subseteq \ER$.
    Again, using oracle separation, we could show that this is not likely.
    We show an even stronger statement by comparing \BQP with \RPH.

    \begin{restatable}{theorem}{BQPnotinRPH}
    \label{thm:BQPnotinRPH}
        There exists an oracle $O$ such that $\BQP^O \nsubseteq \RPH^O$.
    \end{restatable}

    \paragraph{The main message.}
    The main message of our study is that working with a real model of computation \textit{seems not} to give any great additional power to the binary counterparts.
    This confirms our intuition that most problems can be solved more or less equally fast on the wordRAM and the realRAM.
    This is in contrast to previous models, which appear significantly more powerful.
    For example, the simple addition of rounding allows a real Turing machine to compute $\PSPACE$-complete problems in polynomial time.
    So, despite some models that work with real numbers being extraordinarily powerful, we convey that simply using the realRAM is not an all-powerful model that can solve problems unrealistically fast.


    \paragraph{Foundations.}
    In order to have a self-contained work, we give a detailed account with elaborate definitions of our model of real computation and other basic concepts, see \Cref{sec:preliminaries}.


\subsection{Polynomial Time Computations with Reals}
\label{sub:RealComputation}

This article addresses the power of real number computations, which makes it worthwhile to provide some background on their practical applications. One might ask: if real number computations are rarely used, why should we care? The answer is that real number computation is ubiquitous, used across many domains, making it difficult to comprehensively list every application. Instead, we provide a few selected examples.  

\begin{itemize}
\item \textbf{Computational Geometry:} Computational Geometry often uses the realRAM model as its standard model of computation, primarily due to the inherent nature of geometric data being real-valued. Additionally, computational geometers prefer not to deal with numerical issues that arise from rational number computations, nor with challenges such as collinearities and degenerate inputs. Although these are practical issues, addressing them is cumbersome and often considered uninteresting, leading researchers to ignore them in theoretical models~\cite{mark2008computational, de2000computational}.
\item \textbf{Convex Optimization:} In Convex Optimization, the goal is to find a point that minimizes a linear objective function given a convex program. According to O'Donnell~\cite{NotAutomatizable}, the community often assumes that semi-definite programming (SDP) can be solved efficiently, i.e., in polynomial time. However, as pointed out in O'Donnell's work, this is not always the case, particularly in specific degenerate scenarios.
\item \textbf{Continuous Probabilities:} In areas such as machine learning, statistics, and smoothed \mbox{analysis}, it is common to draw a random number $x$ from the interval $[0,1]$. Without the ability to handle real-valued inputs, further computation on this value would be impossible. Thus, these fields implicitly require computational models that support real numbers. A well-known example is the smoothed analysis framework by Spielman and Teng~\cite{spielman2009smoothed}.
\end{itemize}

\subsection{The Existential Theory of the Reals}
\label{sec:ETR}

\paragraph{Significance.}
The complexity class \ER or also called the existential theory of the reals gained popularity in the last decades due to the fact that a series of central problems in theoretical computer science became known to be \ER-complete.
Thus \ER-completeness is describing the complexity of those problems precisely.
The lecture notes by Matou\v{s}ek~\cite{M14} gives a gentle introduction to the complexity class.
The compendium by Schaefer, Cardinal and Miltzow~\cite{ERcompendium} gives an overview of all known \ER-complete problems.
Some examples are stretchability of pseudoline arrangements~\cite{S91,M88}, non-negative matrix factorization~\cite{S16}, Nash equilibria~\cite{DGP09, SS17}, training neural networks~\cite{BHJMW22, Z92,AKM21}, polytope realization~\cite{RG99} and geometric packing~\cite{AMS24}, to name just a few highlights.

\paragraph{History.}
The name appeared first in a conference version by Schaefer~\cite{S10}, yet a handful of \ER-completeness proofs 
had already been established before, see~\cite{S91,RG99, B91, KM94}.
The complexity class under the name $\BP^0(\NP_\R)$ was already used considerably earlier, but the connection between the two research lines had only been made later~\cite{ERcompendium}.
While Schaefer defined \ER using logical formulas over the reals, 
$\BP^0(\NP_\R)$ is defined using real Turing machines. 
The equivalence between the two is not difficult to establish.
The definition via real Turing machines allows us to ask about oracle separations, and so we use this definition here as well. 
In contrast, the definition via logic, makes it convenient to be used as a starting point for hardness reductions.

\paragraph{Practical Difficulty.}
    In order to understand how meaningful \ER-completeness vs \NP-completeness is, practical difficulty is an important point to consider.
    The vague impression from the authors is that \ER-complete problems are more difficult to solve than \NP-complete problems. 
    We will give arguments that support this perspective.
    However, one should not be aware that it is possible to create \NP-complete problems for which small generic instances are impossible to solve.
    Furthermore, there are \ER-complete problems which are solved sufficiently in practice with giant instances, like training neural networks.
    Nevertheless, we want to argue that in some sense \ER-complete problems are ``typically'' harder to solve than \NP-complete problems.
    We give a few reasons below.
    \begin{itemize}
        \item There is no ``simple'' or straightforward algorithm to solve \ER-complete problems optimally.
        All existing algorithms require deeper knowledge of polynomials and real algebraic geometry.
        \item There are general purpose algorithms to solve \NP-complete problems. The two most famous ones are Integer Linear Programming solvers and SAT-solvers.
        Both of them run very fast in practice on huge instances with provable optimal solutions.
        \item Fast methods for \ER-complete problems are typically based on gradient descent method (sometimes in disguise). 
        Gradient descent has no guarantee on the optimum and could stop any time in a local optima arbitrarily far away from the global optimum.
        This behavior is avoided in training neural networks by carefully designing neural network architectures, with redundancies, data diversity and special network structures, and many other clever tricks that developed over decades from a very large and very active community. 
        \item There are tiny instances known of \ER-complete problems that are not solved by humanity.
        The prime example is finding the smallest square container to fit $11 = 1 + 1 + 1 + 1 + 1 + 1 + 1 + 1 + 1 + 1 + 1$ unit squares.
        In comparison, any SAT instance with $11$ variables can solved by a patient teenager, or by a brute-force program written in 30 minutes that runs in 10 milliseconds on the phone in your pocket.
    \end{itemize}

\paragraph{Complexity.}
As mentioned before, it is known that $\NP \subseteq \ER \subseteq \PSPACE$. 
While the first part follows directly from the definition of \NP and \ER (You need to use the right definition though.), the second part is considered a major breakthrough, first established by Canny~\cite{C88b}, see also Renegear~\cite{R92A, R92B, R92C, R92}.
Our results hint that \ER is strictly contained in \PSPACE.

\paragraph{Beyond \ER.}
We are currently aware of only a handful of problems that are known to be complete for other levels of the polynomial hierarchy. 
Note that the second level $\SRPHlevel{2}$ is also denoted $\exists \forall \R$ and $\PRPHlevel{2}$ is denoted $\forall \exists \R$.

The first problem is a generalization of \ETR.
In \ETR we are given a quantifier free formula $\varphi$
and we ask if there exists $x\in\R^n$ such that $\varphi(x)$ evaluates to true.
This generalization is called $Q_1\ldots Q_k$-\textsc{Fragment of the Theory of the Reals}.
Here $Q_i$ are alternating existential and universal quantifiers.
So this generalization is really a family of algorithmic problems.
In this generalization, we ask if
$Q_1 x_1\in \R^n \ldots Q_k x_k\in \R^n : \varphi(x_1,\ldots,x_k)$.
B\"{u}rgisser and Cucker showed that this generalization is complete for the corresponding level of the real polynomial hierarchy~\cite{BC09}. For example, the $\exists \forall$ fragment is complete for $\exists \forall \R$.
Furthermore, Schaefer and \v{S}tefankovi\v{c}~\cite{SS23}
showed that $\varphi$ can be replaced by the predicate $p>0$, for some multivariate polynomial $p$.
And additionally, instead of quantifying over the reals it is sufficient to quantify over the unit interval.
While this result is a bit abstract and technical, it is an important tool for hardness reductions.

For example, it was a key step to show that computing the \textsc{Hausdorff Distance} is $\forall \exists \R$-complete~\cite{JKM23}.
Other problems are the \textsc{Compact Escape Problem}~\cite{DCLNOW21}, \textsc{Surjectivity}~\cite{SS23, BC09}, \textsc{Image Density}~\cite{BC09, JJ23}, 
being \textsc{Star Shaped} for compact semi-algebraic sets~\cite{starhard05},
and a restricted version of \textsc{Area Universality}~\cite{DKMR18}.

\subsection{Oracle Separations}
\label{sub:OracleHistory}

We start with the historical motivation of oracle separation results.
In \Cref{tab:oracle}, we give a set of selected of milestone results.

\begin{table}[tbp]
    \centering
\begin{longtable}{|p{4.2cm}|p{10cm}|}
    \hline
    \textbf{Author \& Year} & \textbf{Results} \\
    \hline
    1975 Baker,  Gill and Solovay~\cite{baker1975relativizations}  & Oracle separation of P and \NP \\ 
    \hline
    1981 Furst, Saxe, and Sipser~\cite{FSS84} &  Relation between constant-depth circuits and the polynomial hierarchy \PH
    \\
    \hline
    1981 Bennett and Gill~\cite{bennett1981relative}    & Relative to a random oracle $O$, $\text{P}^{O}\neq \NP^O \neq \text{co-NP}^O$ with probability 1\newline Introduction of the \emph{random oracle hypothesis}: random oracle separation implies unrelativized separation (disproved; see later in the table)\\
    \hline
    1985 Sipser~\cite{Sip83} and Yao~\cite{Yao85}  &  Oracle separation of levels of \PH, for some oracle\\
    \hline
    1986 H\aa stad~\cite{hastad1986} & 
    Influential Switching Lemma to show circuit lower bounds
    \\
    \hline
    1986 Cai~\cite{Cai86} & Relative to a random oracle \(O\), $\PH^O \subsetneq \PSPACE^O$ with probability $1$\\
    \hline
    1994 Book~\cite{book1994collapsing} 
    & Collapsing of the polynomial time hierarchy without oracle access (\PH) implied by collapsing of polynomial time hierarchy with access to a random oracle ($\PH^O$) with probability \(1\)\\
    \hline
     1994 Chang, Chor,  Goldreich,   Hartmanis,   H{\aa}stad,   Ranjan, Rohatgi~\cite{chang1994random}
     & Refutation of the random oracle hypothesis: $\textsc{IP}^O \neq \PSPACE^O$, despite $\textsc{IP} = \PSPACE$~\cite{shamir1992ip}\\
    \hline
    
    2010 Aaronson~\cite{Aaronson2010}  & $\BQP^O$ not contained in $\PH^O$ for some oracle $O$ under additional assumptions \newline Arguments for importance of oracle separations\\
    \hline
    2015 Rossman, Servedio and Tan~\cite{Rossman2015FOCS} 
    & Distinctness of each level of \PH\ relative to a random oracle \(O\) with probability \(1\) \\
    \hline
    2022 Raz and Tal~\cite{RT22}  & $\BQP^O$ not contained in $\PH^O$ for some oracle $O$ \\
    \hline
\end{longtable}
\caption{A chronological list of selected results related to oracle separation.}
    \label{tab:oracle}
\end{table}

\paragraph{Historic Motivation.}
To give a motivation of oracle separation results, it is maybe worth to go one more step back.
The origins of modern complexity theory lie in a formal definition of computations in the 1930s.
Some computational problems are shown to be undecidable~\cite{AB09} using diagonalization arguments.
Furthermore, and particularly of note in the context of computations with real numbers, researchers defined the arithmetical hierarchy~\cite{Kleene43,Mostowski79} that contains different levels of undecidability.
It is known that the different levels of the arithmetical hierarchy are distinct (without the use of oracles)~\cite{Rogers69}.
Later, the focus shifts to the study of more limited complexity classes, like \P, \NP, \PH, \PSPACE etc.
The original hope was that proof strategies that 
could show the undecidability of the \textsc{Entscheidungsproblem} would carry over to those more fine-grained new complexity classes.
Interestingly, all early separation proofs used only simple properties of Turing machines.
It is not exactly clear how to define ``simple'' properties, but to give the intuition, consider two complexity classes $\mathcal{A}$ and $\mathcal{B}$.
Let us assume that they are defined using the two machine models $M_A$ and $M_B$.
Now, alternatively, we could also define the two complexity classes using machine models $M_A^O$ and $M_B^O$ for some oracle $O$.
(This is also called the relativized world with respect to oracle~$O$.)
This leads to the complexity classes $\mathcal{A}^O$ and $\mathcal{B}^O$.
Assume further that $\mathcal{A}^O = \mathcal{B}^O$,
but we expect that $\mathcal{A} \neq \mathcal{B}$.
Then it is clear that we need to use some property of $M_A$ and $M_B$
that is not true for $M_A^O$ and $M_B^O$.
In summary, a ``simple'' property is such a property that is true for all oracle machines $M^O$.

\paragraph{Technical Ideas.}
The main idea that underscores many oracle separation results, is a relation to small depth circuits that was first pointed out by Furst, Saxe and Sipser~\cite{FSS84}.
To explain the underlying idea, we consider the function \Parityn. 
It receives as input $n$ binary numbers $x_1,\ldots,x_n$ and returns the parity of $x_1+\ldots +x_n$.
\Parityn\ can be solved in polynomial time, but it \textit{cannot} be solved by constant-depth circuits of ``small'' size.
The proof of the latter fact is not easy -- one way to establish it is using H\r astad's switching lemma~\cite{hastad1986}.

Given an oracle $O$, we can define the language 
\[\ParityA = \{ 1^n \mid \text{\emph{the number of binary strings of length $n$ in $O$ is odd}}\}.\]
Intuitively, one can see the problem of deciding $1^n \in \ParityA$ as solving $\Parity$ with exponentially many inputs.
Namely, all the results of queries $x \in O$ with $x$ of length $n$.
The key insight is that as $\Parity$ cannot be solved with ``small'' constant depth circuits, neither can \ParityA.
Furthermore, it is possible to translate an algorithm solving problems in $\PH^O$,
as a circuit of constant depth of ``small'' size.
(Quasi-polynomial in terms of the exponential input the circuit receives.)
As the lower bound of the size of the circuit is larger than the upper bound, we 
will conclude that $\ParityA$ is not contained in $\PH^O$.
Now, using that \ParityA is contained in $\PSPACE^O$, we will conclude that
 $\PH^O$ is \textit{properly} contained in $\PSPACE^O$.

Other oracle separation results work similarly.
We summarize the key steps:

\begin{enumerate}
    \item Consider a simple function $f : \{0,1\}^*\rightarrow \{0,1\}^*$. (\Parity was the function we used above.)
    \item Find some circuit lower bound for computing $f$.
    \item Translate this circuit lower bound to the oracle world with exponentially many inputs.
    \item Show that the considered complexity class yields smaller than possible circuits.
\end{enumerate}
Once the lower bound is above the upper bound, we can imply a contradiction and have identified a function outside the considered complexity class.
We follow the general approach and build upon already existing lower bounds.
So, we aim to find appropriate upper bounds for circuits deciding~$\RPH^O$.

\paragraph{Real Oracle Separations.}
To the best of our knowledge there are three papers that deal with oracle separations of real complexity classes.
The first result is by Emerson~\cite{emerson1994relativizations} who showed that $\P_\R$ and $\NP_\R$ can be separated using oracles.
The second result by Ga{\ss}ner~\cite{gassner2010separation} is separating $\P_\R$ from $\text{DNP}_\R$ using oracles.
(In the complexity class $\text{DNP}_\R$ we have a binary witness, but can do computations with real numbers.)
The third result by Meer and Wurm~\cite{MeerWurm2025} shows that 
$\NP^\Z \nsubseteq \ER^\Z \neq \PSPACE^\Z$.
Here the integers $\Z\subset \R$ serve as an oracle.
Note that boolean models of computation cannot make much use of real oracles.

\subsection{Proof Overview}
Similarly as for the corresponding known separation results involving the classic polynomial hierarchy, we can use the following:
\begin{itemize}
    \item \ParityA is in \(\PSPACE^O\) and there is a known lower bound on the size of any constant-depth circuit deciding \ParityA~\cite{Cai86}.
    \item There is a so called \emph{Sipser-like function} for each constant \(k\) which lies in the \(k\)-th layer of the polynomial hierarchy and for which there is a known lower bound on the size of any circuit of depth \(k - 1\)~\cite{Yao85, Sip83}.
    \item There is a distribution over \(\{0,1\}^*\) such that deciding whether a sequence was sampled from this distribution when restricted to be sampled from the uniform distribution otherwise can be done with a certain bound on the probability of false positives in \BQP\ and for which there is a known lower bound on the size of any constant-depth circuit ensuring the same guarantee \cite{RT22}.
\end{itemize}


The crucial difference between known separation results involving the polynomial hierarchy and our new separation results involving the real polynomial hierarchy lies in the following fact:
In the classic polynomial hierarchy the quantification over variables given to a Turing machine in addition to a language \(L\) on a specific level of the hierarchy range over \textit{Boolean} values.
Hence, a bounded-depth circuit for computing \(L\) can be given by encoding the characteristic function of \(L\) on all Boolean input strings of length \(n\), i.e., via \(2^n\)-many inputs.
This is not immediately possible for quantification over \textit{real} variable ranges, which underlies the definition of the real polynomial hierarchy.


The technical basis of our separation results is \Cref{lem:ConstrainedQuantifierRanges} which shows the existence of some set~$\LLL$ which we can use instead of $\R$ as the range for the quantifiers.
This in turn enables us to derive new upper bounds on 
the circuit size of any \RPH algorithm, see \Cref{thm:RPHandACanalogue} and \Cref{thm:CircuitLowerBoundsImplyNoMembershipInRPH}.
And those upper bounds are smaller than known lower bounds for functions in $\PSPACE^O$, $\BQP^O$ as well as higher levels of $\PH^O$.
And in this way we conclude the separation results.
We note that this set~$\LLL$ is not constant or even of constant size -- in fact we give an easy argument, why it cannot be.
The sets to which we can restrict the quantifier ranges will be constructed inductively using a bound on the number of real Turing machines which we establish in \Cref{lem:NrOfConstfreeRTMs}.
\subsection{Discussion}
\label{sec:discussion}

In this section, we aim to discuss our results from different perspectives.
That includes to give a fair evaluation, highlight shortcomings, and possible improvements or alternative perspectives.

\paragraph{How Meaningful are Oracle Separations.}
    As mentioned earlier, oracle separation really separates machine models and not complexity classes.
    How meaningful those separations are with respect to the original machine model is worth a discussion.
    In the following, we list arguments that either support or reduces the strength of oracle separation results, as far as we are aware of them.
    
    \begin{itemize}
        \item 
            At first, it is important to note that oracle separations are not decisive.
            Specifically, it is known that  $\textsc{IP} = \PSPACE$~\cite{shamir1992ip}, but there are oracles \(O\) such that $\textsc{IP}^O \neq \PSPACE^O$~\cite{chang1994random}.
        \item Oracles are still useful, because
        without being able to show \(\mathcal{A}^{O} \neq \mathcal{B}^{O}\) for any oracle \(O\) one cannot even hope to show \(\mathcal{A} \neq \mathcal{B}\) for two complexity classes with ``black box methods'' such as diagonalization.
        For the polynomial hierarchy, in particular if \PH with random oracle access were to collapse, so would \PH without oracle access~\cite{book1994collapsing}.
        
        \item Aaronson~\cite{Aaronson2010} makes the case that oracle separations give ``\textit{lower bounds in a concrete computational model that is natural and well motivated}''.
        Aaronson refers to this model as \textit{query complexity}. 
        The idea is that we lower bound the number of queries to the oracle.
        Maybe, instead of the number of queries, one might also want to think of the way the results from the queries are processed.
        In our work, we use known lower bounds and give new upper bounds.
        \item At last, 
        whenever we have some oracle separation results,
        they are based on some unconditional separation results in some concrete models of computation.

    \end{itemize}

\paragraph{Random Oracle Separation.}
    While \Cref{thm:PSPACEnotRPH} is true for a random oracle,
    the oracles for \Cref{thm:RPHlayers} and \Cref{thm:BQPnotinRPH} are of a more complicated structure. 
    We wonder if the later separations are also true for random oracles.
    There are several reasons why we care about random oracles more than just a specific oracle.
    \begin{itemize}
        \item One reason is that a random oracle seems to capture on an intuitive level more an ``arbitrary'' oracle in comparison to a well-crafted oracle.
        The specific oracle could be engineered to make a certain proof work, but has maybe less chance to give us an intuition to the world without oracles.
        If we can show that two complexity classes are different with respect to a random oracle with probability $1$,
        then we showed that they are different with infinitely more oracles than just one.
        And thus this gives a stronger intuition that the complexity classes are indeed different without oracles as well.
        Indeed, the \textit{random oracle hypothesis} states that if two complexity classes are different with respect to a random oracle, then they are different~\cite{bennett1981relative}.

        We want to warn here to over interpret this idea.
        For starters, $\textsc{IP}^O \neq \PSPACE^O$, even for a random oracle $O$.
        In other words, the random oracle conjecture is disproven.
        Furthermore, random structures all ``look the same''.
        To understand what we mean, let $P$ be a property and $S$ be a random structure. 
        We know that $P$ will hold for $S$ with either probability $0$, or probability $1$, when the size of $S$ goes to infinity.
        So, if we are interested with oracles with a different set of properties, a random oracle, gives us only \textit{one} example.

        So arguably more interestingly would it be to have oracle separations with different types of oracles, i.e., oracles satisfying different types of properties.
        
        \item Another reason that ``we'' care about random oracles is that several members of the research community care about it~\cite{barak_ac0_notes, hastad1986, Rossman2015FOCS}.
    \end{itemize}

\paragraph{PosSLP Separations.}
    The maybe most notable work to describe the difference between binary and real computation is via \PosSLP.
    In \PosSLP we are given a finite sequence $a_1,\ldots,a_n$ as follows.
    $a_0 = 0, a_1 = 1$ and each $a_k = a_i \oplus a_j$, with $i,j <k $ and $\oplus \in \{+,-,\cdot\}$.
    Note that this sequence is not given using binary numbers, but by merely giving the sequence of operations to compute each $a_k$.
    \PosSLP asks if the number $a_n$ computed in that way will be positive.
    Note that \PosSLP can be encoded using the number $0,1$.

    It was shown by Allender, B\"{u}rgisser, 
    Kjeldgaard-Pedersen, and Miltersen~\cite{ABKPM08}
    that the polynomial-time computation on the wordRAM with an oracle to \PosSLP is equally powerful as polynomial-time computation on the realRAM.
    (Both under binary inputs.)
    Thus, in that sense, \PosSLP captures the difference between the real computation and binary computation with regard to polynomial-time algorithms.
    In other words $\P^{\PosSLP} = \BP(\P_\R) = \PR$.
    Note that this has no direct consequences to the higher levels of the (real) polynomial hierarchy.
    The reason is that guessing real numbers and doing computations with real numbers is not the same.
    
    B\"{u}rgisser and Jindal~\cite{BJ24} showed that \PosSLP is NP-hard under some strong assumptions.
    As $\P^{\PosSLP} = \BP(\P_\R)$ this would imply that $\NP \subseteq \PR$.
    Now, \Cref{cor:NP-notin-PR}
        says that there exists a
        binary oracle $O$ such that $\NP^O \nsubseteq \PR^O$. 
    Thus, if we believe the underlying assumptions from
    B\"{u}rgisser and Jindal~\cite{BJ24}, then they have found a property of polynomial time Turing machines that are not true for oracle Turing machines.

\paragraph{Separating \NP from \ER.}
    Maybe one of the most interesting questions is if we can give an oracle separation between \NP and \ER. 
    Note that this is not possible with our techniques, as \ER is the real analogue of \NP, 
    so we cannot build on previous oracle separation results.
    Some authors conjectured that $\NP^\PosSLP = \ER$~\cite{ERcompendium}.
    If true, then \PosSLP would explain the difference between \NP and \ER.
    Yet, it is possible that there are suitable oracles to separate \NP from \ER without fully understanding the difference between \NP and \ER.


\paragraph{Smoothed Analysis.}
    We are aware of a structural approach to limit the power of real RAM computations.
    Erickson, Hoog and Miltzow showed that a host of \ER-complete problems can be solved in \NP time with high probability, if we apply a small random perturbation to the input~\cite{EvdHM20}.
    Although the methodology of smoothed analysis is completely different to oracle separation, the aim is similar. 
    The aim is to limit the power of computations with real numbers.

\paragraph{Notation and Motivation of Model.}
    We have taken the liberty to deviate from previous naming of our complexity classes.
    In the BSS model, \ER is denoted by $\BP(\NP^0_\R)$,
    and \PR is denoted by $\BP(\P^0_\R)$~\cite{BSS89,BCSS98}.
    The abbreviation $\BP$ stands for binary part and the zero ``$0$'' indicates that no constants are added to the machine.
    We use our notation and thus deliberately deviate from previous notation in the literature.
    The reason being that to us the complexity classes without machine constants and with binary inputs are the more natural ones to study and thus deserve a catchy name that is easy to parse and to remember.
    There are several reasons we found our complexity classes more interesting. 
    \begin{itemize}
        \item Introducing arbitrary real constants that the machine can use leads to power that is at least \P/poly~\cite{Koiran94}.
        \item We know many more complete problems for \ER       than for $\NP_\R$~\cite{ERcompendium}. 
        \item When we restrict to deciding binary languages, we can compare those complexity classes to classical binary complexity classes. 
    \end{itemize}
    
For good measure, it is fair to mention that certain results are more natural when formulated using real inputs. 
Furthermore, even if we only care about binary inputs, some results require to talk about complexity classes with real inputs~\cite{BCSS98}.

\section{A Constant-Depth Circuit Analogue for $\RPH$}
\label{sec:ACanalogueForRPH}

In this section, we show how the real polynomial hierarchy relates to constant-depth circuits. The computations of any $k$-level $\RPH$ machine (on a fixed input length) can be translated into a $k$-depth Boolean circuit by first replacing every $\forall$ quantifier with an $\land$ gate and replacing every $\exists$ quantifier with an $\lor$ gate and then using the outputs of the Turing machine on all possible witnesses as inputs feeding into this circuit. 
However, such a straightforward translation would require that each $\land$ and $\lor$ gate in the resulting circuit has infinite fan-in (because our witnesses can be any vector in $\mathbb{R}^m$); see \Cref{fig:IllustrationSimple} below for example.

\begin{figure}[h]
\centering
    \includegraphics[]{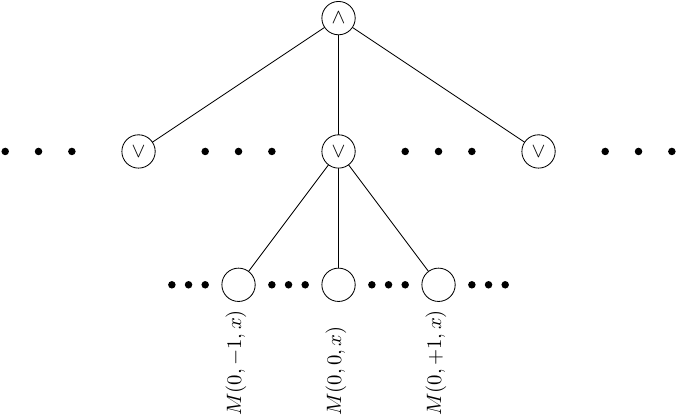}
    \caption{Illustration of a circuit with infinite fan-in corresponding to $\forall u \exists v : M(u,v,x) = 1$.}
\label{fig:IllustrationSimple}
\end{figure}

Note that the first level (i.e., the root) corresponds to the universal quantifier and the second level corresponds to the existential quantifier.
Fortunately, we are able to argue that this fan-in can be made finite and bounded. 
We do this in three steps.
First, we give an upper bound on the number of different real Turing machines running within a certain time bound a certain number of computational steps.
Then, we show how we can change the set that we quantify over from the infinite set $\R^m$ to some finite set $\LLL \subseteq \R^m$.
The underlying idea is that there are only finitely many different machines, for our scenario, and we only need one element of $\R^m$ for each different  machine.
In the third part, we use the upper bound on the cardinality of $\LLL$ to give a circuit recognizing any language of the real polynomial hierarchy with exponential size circuits.

\subsection{On the Number of Real Turing Machines}
\label{sub:numberTM}
Later, it will be useful to upper-bound the number of real Turing machines that behave differently in terms of which real numbers on their tape are relevant for their computations on inputs of a given size.
For this we use the definition of (oracle) real Turing machines as presented in \Cref{def:RealTuringMachine,def:OracleRTM}.
While \Cref{lem:NrRTMsinstates} supplies a bound on the number of real Turing machines that depends on their number of states, the unnecessary inclusion of states which are not relevant for any computation path makes this bound arbitrarily bad.
For this reason, we define the following notion of equivalence between real Turing machines.
\begin{definition}[Real Oracle Turing machine equivalence]
    For a number \(n \in \mathbb{N}\) and a function $T : \mathbb{N} \rightarrow \mathbb{N}$, real oracle Turing machines \((M_1,O_1)\) and \((M_2,O_2)\) are \emph{$T(n)$-equivalent} if and only if on any input of size $n$ both machines terminate after $T(n)$ computational steps and the content on the tapes of \(M_1\) and \(M_2\) are equal after each computation step.
    Two real oracle machines are \emph{$T(n)$-distinct}, if they are not \emph{$T(n)$-equivalent}.
\end{definition}

\begin{lemma}
\label{lem:NrOfConstfreeRTMs}
    Let \(n \in \mathbb{N}\) and let $T: \mathbb{N} \rightarrow \mathbb{N}$ be a function. 
    Then there are at most \(4860^{T(n)}\) many $T(n)$-distinct real oracle Turing machines with the
    following property.
    For any input $x \in \R^n$ it runs in $T(n)$ steps.
    Here, two real oracle Turing machines are considered different when they receive a different oracle results, even if the computational steps are the same.
    
\end{lemma}

\begin{proof}
    Two Turing machines are potentially not $T(n)$-equivalent if they do at least one of the three things differently:
    \begin{enumerate}
        \item a different head movement, ($3^5 = 243$ different possible head movements.)
        \item a different binary operation ($2$ different binary operations), or
        \item a different real or oracle operation (at most $13$ different real operations).
    \end{enumerate}
    Note that if two Turing machines, have the same input and do the same operations and the same head movements, then the tape content will be identical at all time steps.
    Thus, per time step, there are at most $4860$ choices by the real Turing machine that affect the tape or the head movement.
    As there are at most $T(n)$ time steps to be considered,
    we get at most 
    \(4860^{T(n)}\) real Turing machines that are not $T(n)$-equivalent.
\end{proof}

\subsection{Quantifying over Finite Sets}
\label{sub:finitequantification}
In the next lemma, we show how the infinite set $\R^m$ can be replaced by a finite set in the quantification.
To make our life slightly easier, we will not replace $\R$ but the entire set $\R^m$ by a set of vectors $\LLL \subseteq \R^m$.
Intuitively, the lemma argues that the finite number of Turing machines implies the existence of a finite set of solutions, without changing the truth value of any sentence.
The core idea is that we only need one vector per machine, and the total number of machines with a specific running time is finite. Therefore the total number of vectors we need to consider is finite.
Let us mention that the existence of the set $\LLL$ does not imply we can compute it quickly.

\begin{lemma}
    \label{lem:ConstrainedQuantifierRanges}
    Let \(k \in \N\), $q:\N \rightarrow \N$ be a polynomial, $T:\N \rightarrow \N$ be a function satisfying $T(n) \geq k\cdot q(n)+n$ and $1 \leq i \leq k$. 
Then for every $n\in \N$, there is a set \(\LLL(n, T,i, k) \subseteq \R^{q(n)}\) such that for every real Turing machine \(M\) with runtime at most $T(n)$, every $x \in \{0,1\}^n$ and every oracle $O$ the following holds.  
(1) The cardinality of \(\LLL(n, T,i, k)\) is at most \( (2^n \cdot \multiplier^{T(n')})^{2^{i-1}} \) (where $n'=k\cdot q(n)+n$), and (2) the following two sentences are equivalent:
    \[ Q_1 u_1 \in \R^{q(n)} \ldots Q_k u_k \in  \R^{q(n)}  : M^O(x, u_1,\ldots,u_k) =1,\]
    \[ Q_1 u_1 \in \LLL(n, T, 1,k) \ldots Q_k u_k \in \LLL(n, T, k,k) : M^O(x, u_1,\ldots,u_k) =1.\]
Here $Q_j$ alternates between being $\forall$ and $\exists$.
    
\end{lemma}

\begin{remark}
The sets $\LLL(\cdots)$ also implicitly depends on the polynomial $q$, however as our function $T$ also depends on the polynomial $q$ we avoid using $q$ as an argument in defining $\LLL$.
Additionally, note that the set $\LLL(\cdots)$ is a set of vectors of length determined by $q$. 
\end{remark}

Note that the following proof is slightly easier for the case $k=1$, but as the general case is not much more complicated, we do not write out the case $k=1$ separately.

\begin{proof}
    Let $n'=k \cdot q(n) +n$ denote the length of the entire input to the Turing machine $M^O$.
    We do the proof by induction on $i = 1, \ldots,k$.
    We first do the induction basis for $i=1$.
    
    \paragraph{Induction Basis $i=1$.} 
    There are at most  $\multiplier^{T(n')}$  real oracle Turing machines $(M,O)$ that are not $T(n')$-equivalent; see \Cref{lem:NrOfConstfreeRTMs}.
    For each of them, the statement 
    \[\exists u_1 \in \R^{q(n)} (\forall u_2 \in \R^{q(n)} \ldots Q_k u_k \in \R^{q(n)}): M^O(x, u_1,\ldots,u_n) = 1\]
    is either true or not.
    If it is true, then we add the corresponding $u_1 \in \R^{q(n)}$ to $\LLL(n,T,1,k)$. We do this for all $x \in \{0,1\}^n$.
    Similarly, for each $(M,O)$ for which the statement 
    \[\forall u_1 \in \R^{q(n)}  (\exists u_2 \in \R^{q(n)} \ldots Q_k u_k \in \R^{q(n)}): M^O(x, u_1,\ldots,u_k) = 1\]
    is false, we add real numbers to $\LLL(n,T,1,k)$ that witness that the statement is incorrect and again, we do this for all $x \in \{0,1\}^n$.
    We note that $|\LLL(n,T,1,k)| \leq 2^n  \cdot \multiplier^{T(n')}$; 
    for any $x$, we either add a certificate for the $\forall$ or the $\exists$ quantifier depending on the truth value of $M^O$ deciding $x$.
    Clearly, we can replace $\R^{q(n)}$ by $\LLL(n,T,1,k)$ for the first quantifier by construction.

    
    \paragraph{Induction Step $1$ to $2$.}
    As the induction step is a little notation-heavy, we do the induction step first from $1$ to $2$, and then in full generality.
    Readers are invited to skip this paragraph.
    Alternatively, it is probably enough to read this paragraph, and then it should be clear how it works in full generality.
    
    First, consider a sentence of the form
    \[\exists u_1 \in \R^{q(n)} \, \forall u_2\in \R^{q(n)}  \ldots Q_k u_k \in \R^{q(n)}: M^O(x, u_1,\ldots,u_k) = 1,\]
    By the induction hypothesis, the statement is equivalent to
    \[\exists u_1 \in \LLL(n,T,1,k)  \, \forall u_2 \in \R^{q(n)}  \ldots  Q_k u_k \in \R^{q(n)} : M^O(x, u_1,\ldots,u_k) = 1,\]
    Now, we construct the set $\LLL(n,T,2,k)$.
    We add one $u_2$ for each $u_1\in \LLL(n,T,1,k)$,
    and $(M,O)$ that is not $T(n)$-equivalent.
    We add to  $\LLL(n,T,2,k)$ a vector $u_2\in \R^n$,
    such that \[\exists u_3 \in \R^{q(n)} \forall u_4\in \R^{q(n)} \ldots  Q_k u_k \in \R^{q(n)} :M^O(x,u_1,\ldots,u_k) \neq 1.\]
    If no such elements exist, then we do not add anything to $\LLL(n,T,2,k)$.
    We do the same for the other order of quantifiers.
    We note that \[|\LLL(n,T,2,k)| \leq 2^n \cdot |\LLL(n,T,1,k)| \cdot \multiplier^{T(n')} \leq 2^{2n} \cdot \multiplier^{2T(n')} = (2^n \cdot \multiplier^{T(n')})^{2^{2-1}}\]
    \paragraph{Induction Step $i$ to $i+1$.}  
    First, consider a sentence of the form
    \[Q_1 u_1 \in \R^{q(n)} 
    \ \cdots \ 
    Q_i u_i \in \R^{q(n)}
    Q_{i+1} u_{i+1} \in \R^{q(n)}
    \ \cdots \ 
    Q_k u_k \in  \R^{q(n)} M^O(x, u_1,\ldots,u_k) =1\]
    where the quantifiers $Q$ alternate.
    By the induction hypothesis, the statement is equivalent to
    \[Q_1 u_1 \in \LLL(n,T,1,k)
    \ \cdots \ 
    Q_i u_i \in \LLL(n,T,i,k)
    Q_{i+1} u_{i+1} \in \R^n
    \ \cdots \ 
    Q_k u_k \in  \R^{q(n)} M^O(x, u_1,\ldots,u_k) =1\]
    
    Now, we construct the set $\LLL(n,T,i+1,k)$.
    We add one $u_{i+1} \in \R^{q(n)}$ for each $(u_1,\ldots,u_i)\in \LLL(n,T,1,k)\times \ldots \times \LLL(n,T,i,k)$,
    and $(M,O)$ that is not $T(n)$-equivalent as follows.
    Let us assume that $Q_i = \exists$.
    We add to  $\LLL(n,T,i+1,k)$ a vector $u_{i+1}\in \R^{q(n)}$,
    such that \[\varphi(u_1,\ldots,u_{i+1}) \equiv Q_{i+2} u_{i+2} \in \R^{q(n)}  \ldots  Q_k u_k \in \R^n :M^O(x,u_1,\ldots,u_n) = 1,\]
    is true.
    If no such elements exist, then we do not add anything to $\LLL(n,T,i+1,k)$.
    In case that $Q_{i+1} = \forall$, we add an element in case that 
    \[\varphi(u_1,\ldots,u_{i+1}) \equiv Q_{i+2} u_{i+2} \in \R^n  \ldots  Q_k u_k \in \R^{q(n)} :M^O(x, u_1,\ldots,u_n) \neq 1,\]
    is true.
    We note that 

    \begin{align*}
    |\LLL(n,T,i+1,k)| 
    & \leq 2^n\cdot \multiplier^{T(n')} \cdot |\LLL(n,T,1,k)|\cdot \ldots \cdot |\LLL(n,T,i,k)| \\
    & \leq 2^n\cdot \multiplier^{T(n')} \cdot \Pi_{j=1}^i (2^n\cdot\multiplier^{T(n')})^{2^{j-1}}\\
    & =(2^n \cdot \multiplier^{T(n')})^{2^{(i+1)-1}}
    \end{align*}
    
\end{proof}

We now provide two examples to motivate our decisions in constructing the set $\LLL$.

\begin{example}
    This example shows that we cannot make the upper bound on the number of $T(n)$-equivalent Turing Machines independent of the running time of the real oracle Turing machine.
    For that, we define an infinite sequence of real Turing machines $(M_i)_{i\in \N}$.
    Each real Turing machine that has constant running time also has the following properties:
    \begin{itemize}
        \item the runtime is $\Theta(i)$ and thus independent of the input size,
        \item it returns \texttt{yes} on input $x\geq i$, and otherwise \texttt{no}.
        (We encode $x \in \R$ as a real number and thus has input length $1$.)
    \end{itemize}
    It is easy to construct such a series of Turing machines.
    Now consider the sentence
    \[\exists x \in \R : M_i(x) = 1.\]
    Clearly, the sentence is correct for all $i \in \N$, but for any finite set $\LLL\subset \R$, the following sentence is incorrect:
    \[\exists x \in \LLL : M_i(x) = 1,\]
    for some $i > \max \LLL$.
    As we have shown, by definition all the machines $M_i$ have constant asymptotic running time.
    Thus, if we want to replace the set that we quantify over by a finite set, then this set has to depend 
    on the concrete running time of the considered machines, and not just on their asymptotic running time.
\end{example}

The next example shows that the sets for the different blocks of quantifiers have to be different.
\begin{example}
    We consider a machine $M$ that looks at its first two real inputs $x,y$ and 
    returns yes if $x\geq y$ and no otherwise.
    Consider a sentence of the form 
    \[\exists x\in \R \,\forall y\in \R : M(x,y) = 1.\]
    Clearly this sentence is wrong, because we can choose $y = x+1$ in the second quantification.
    Now consider any finite set $\LLL \subseteq \R$.
    Clearly the following sentence is true:
    \[\exists x\in \LLL \, \forall y\in \LLL : M(x,y) = 1.\]
    To see that it is true, let $x:= \max \LLL$.
    This example shows that each block of quantifiers needs a different set $\LLL$.
\end{example}


\subsection{Circuit Upper Bounds for \RPH}
\label{sub:CircuitUpperBounds}

We are now ready to give an upper bound on constant-depth circuits that decide languages in the real polynomial time hierarchy.
To achieve this, we collect the results of $M^O(x, \dots)$ for all possible choices of the $u_i$ determined by \Cref{lem:ConstrainedQuantifierRanges} and then check via the circuit $C_n$ whether the given quantifier structure is satisfied. 
It is a standard technique to restrict the circuit's depth by the quantifier rank of the formula. Therefore, readers familiar with this technique may want to skip this section, which is written with a novice reader in mind.
Before, we can start, we give some definitions and notation that simplifies to formulate our theorem.

The following definition spells out what it means for a  real oracle machine to \textit{establish} that a language is on the $k$th level of the real polynomial hierarchy with respect to some oracle.
Thereafter, we give an upper bound on the set $\LLL(n,\T,1,k)\times \ldots \times \LLL(n,\T,k,k)$ in order to state the circuit upper bound in \Cref{thm:RPHandACanalogue} and refine it in \Cref{thm:CircuitLowerBoundsImplyNoMembershipInRPH}.

\begin{definition}[Certification of Level Membership]
\label{def:CertificationLevelMembership}
    We say that a real oracle Turing machine $M^O$ with polynomial $\T: \N \rightarrow \N$ running time, and a polynomial $q$
\textit{establishes} 
that a language $L$ is in
$\SRPHlevel{k}^O \cup \PRPHlevel{k}^O$
if
the following sentence holds for any $x\in \{0,1\}^*$.
\begin{equation*}
    x \in L \iff Q_1 u_1 \in \R^{q(|x|)} Q_2 u_2 \in \R^{q(|x|)} \cdots Q_k u_k \in \R^{q(|x|)} M^{O}(x, u_1,\ldots,u_k)=1.
\end{equation*}
Here $Q_i$ alternates between $\exists$ and $\forall$.
\end{definition}


We denote $\LLL = \LLL(n,\T,1,k)\times \ldots \times \LLL(n,\T,k,k)$, where $\LLL(n, \T, i, k)$ denotes the sets constructed in \Cref{lem:ConstrainedQuantifierRanges}
and $\RPHcircuitUpperBound = |\LLL{}|$ is its size.
Using the definition of $\LLL$ and the known upper bounds on $\LLL(n,\T,i,k)$, we get
\[\RPHcircuitUpperBound = \Pi_{j=1}^{k} | \LLL(n,\T,j,k)| \leq \Pi_{j=1}^{k} (2^n \cdot \multiplier^{\T(n')})^{2^{j-1}} = (2^{n} \cdot \multiplier^{\T(n')})^{2^{k}}
= (2^{2^k n} \cdot \multiplier^{\T(k\cdot q(n)+n) 2^{k}})
.\] 
Note that 
\[S \leq 2^{2^k\bigO( n^c)}, \]
for some constant $c$ depending on $q$ and $\T$.
As $q$ and $\T$ depend only on $M$, so does the constant $c$.
Furthermore, we use the notation $[S]$ for the set $\{1,\ldots, S\}$.

\begin{definition}[Truth Table corresponding to Oracle $O$ and Integer $n$] 
\label{def:TruthTableOfO}
Let $O$ be an oracle. For any string $y\in \{0,1\}^*$, let $b_y \in \{0,1\}$ indicate whether a string $y \in O$, i.e.\ $b_y=1$ iff $y\in O$. The \emph{truth table of $O$} for a value $n \in \N$, denoted by $\truthtable(O,n)$, 
is a Boolean string of length $2^n$,
i.e., $\truthtable(O,n) \in \{0,1\}^{2^n}$.
The string $\truthtable(O,n)$ is
constructed in the following way:
\begin{equation*}
    \truthtable(O,n)=\bigcirc_{y \in \{0,1\}^n} b_y,
\end{equation*}
here $\bigcirc$ denotes the symbol for concatenation and $y$ is picked in the lexicographical order.
\end{definition}

For example $\truthtable(O,3) = 
b_{000}b_{001}b_{010}b_{011}b_{100}b_{101}b_{110}b_{111}$ is a string of length $8$.

\begin{definition}[Oracle Space]
    We denote by $\mathcal{O}$  the family of all binary oracles.
    In other words the \textit{space of all binary oracles}.
    Note that each oracle $O$ can be identified with an infinite bit string $(b_i)_{i\in \N}$.
    The idea is that we can enumerate all the words in $\{0,1\}^*$.
    If $w$ is the $i$th word in this order we let $b_i =1$  if it is in $O$ and $b_i = 0$
    if $w$ is not in $O$. 
    This gives a one to one correspondence between the set of infinite bit strings $\{0,1\}^\N$ and the set of all oracles $\mathcal{O}$. 
    We show now how to define probabilities using this correspondence.
    Let $\sigma \in \{0, 1\}^{n}$ be a string on $n$ bits.
    A \textit{basic cylinder} $[\sigma] \subset \mathcal{O}$ is a subset of oracles such that the following holds.
    Let $O$ be an oracle and $(b_i)_{i\in \N}$  be the corresponding infinite bit string. 
    Then $O \in [\sigma]$ if and only if $b_1\ldots b_n = \sigma$.
    The probability of a random oracle lying in $[\sigma]$ is $2^{-n}$.
    This defines a probability space $(\mathcal{O}, \Pr)$.
    Note that every word is present in a random oracle with probability $1/2$.    
\end{definition}

As the definition is a bit technical, we give a small example.
In lexicographic order, $0$ is the first word in $\{0,1\}^*$ and $1$ is the second word and $00$ is the third word and so on.
Assume we only want to consider the set of oracles $\mathcal{O}_1$ containing both $0$ and $00$, but not the word $1$ then we define $\sigma = 101$.
The set of oracles $\mathcal{O}_1$ equals the basic cylinder $[\sigma]$.
The probability of this space equals $1/8$.
We leave it to the reader to verify that $\PR ([01] \cup [10]) = 1/2$.
Note that $[01] \cup [10]$ describes all oracles that contain exactly one word of length one.

\begin{definition}[Density]
We denote for an infinite string $x\in \{0,1\}^\N$ the first $n$ bits $x_1\ldots x_n$, by 
$x{\upharpoonright}n$.
Thus $[x{\upharpoonright}n]$ is the basic cylinder of the first $n$ bits of $x$.
For $A \subseteq \{0,1\}^\N$ measurable and $x \in \{0,1\}^\N$, we define the 
\textit{cylinder density} at scale $n$:
\[
D_n(A,x) = \frac{\Pr\!\big(A \cap [\,x{\upharpoonright}n\,]\big)}
                {\Pr\!\big([\,x{\upharpoonright}n\,]\big)}.
\]
\end{definition}

Let $\prop:\{0,1\}^* \rightarrow \{0,1\}$ be a function and $O$ an oracle then we define
    \[ L(\prop, O) =
    \{1^n \mid (\prop(\truthtable(O,n)) =1) \land (n \in \N) \}\]
be a unary language.
The following lemma establishes that we can assume that we can in some specific sense assume that we can restrict ourselves to a single Turing machine to establish that $L(\prop,O)$ is in $\RPH^O$, even if we consider different oracles.

\begin{theorem}[Lebesgue Density Theorem on $\{0,1\}^{\mathbb{N}}$ {\cite[Theorem~5.3.1]{Nygard_2016}}]
    Let $A \subseteq \{0,1\}^{\mathbb{N}}$ be measurable.
    Then for almost every $b \in A$,
    \[
    \lim_{n \to \infty} 
    \frac{\Pr(A \cap [b{\upharpoonright}n])}{\Pr([b{\upharpoonright}n])} = 1.
    \]
\end{theorem}

We will now translate this theorem to the language of oracles.
As preparation we need the following definition.
We define $[O\upharpoonright n]$  as the set of oracles $O'$ for which it holds that among the first $n$
words of are in $O'$ if and only if they are in $O$.
If $b$ is the bit string corresponding to $O$ then $[b{\upharpoonright}n]$ corresponds to $[O{\upharpoonright}n]$

\begin{theorem}[Lebesgue Density Theorem for the Oracle Space]
    Let $\mathcal{O}'$ be a measurable set of oracles.
    Then for almost every oracle $O \in \mathcal{O}'$,

     \[ \lim_{n \rightarrow \infty} \frac{\Pr(\mathcal{O}'\cap [O{\upharpoonright}n])}{\Pr([O{\upharpoonright}n])} = 1.\]
\end{theorem}

\begin{lemma}
\label{lem:oneTM}
    Let $\prop:\{0,1\}^* \rightarrow \{0,1\}$ be a function. 
    And let us assume that with probability $\alpha > 0$, for a random binary oracle $O\in \mathcal{O}$, $L(\prop,O)$ is in $\DRPHlevel{k}^{O}$.

    Then, there exists one real oracle Turing machine $M$ such that 
    $M^{O}$ establishes $L(\prop, O)$ is in $\DRPHlevel{k}^{O}$
    with probability  $\beta > 1-\varepsilon$, for any fixed $\varepsilon > 0$.
    %
\end{lemma}

Before we dwell into the proof, it might be worthwhile to note that the number of oracles is uncountable. There are at least two ways to see this. 
The first one is to note that the power set of a countable set is uncountable.
The second one is easy to see by giving an injective mapping from $[0,1)$ to $\mathcal{O}$.
Let $x \in [0,1)$ and we will define an oracle $O_x$.
Let $0.b_1b_2b_3\ldots$ be the (lexicographically smallest) binary expansion of $x$. 
For each $n$ add the word $b_1\ldots b_n$ to our target oracle $O_x$. 
Note that for any two numbers $x$ and $y$ the two oracles $O_x$ and $O_y$ are distinct.

\begin{proof}  
    We will do the proof for $\varepsilon = 0.01$ for readability.
    We trust that the reader is able to adapt the proof to work for general $\varepsilon$.
    We index all the oracles in $\mathcal{O}$ by the (uncountable) index set $I$.
    In other words, $\mathcal{O} = \{O_i : i\in I\}$.
    For each oracle such that $L(\prop,O_i) \in \DRPHlevel{k}^{O_i}$, 
    we denote by $M_i$ the Turing machine establishing this fact.
    If $L(\prop,O_i) \not \in \DRPHlevel{k}^{O_i}$, then $M_i$ is the Turing machine returning no and stopping immediately.
    
    Now, note that we have uncountably many oracles $O \in \mathcal{O}$. 
    Yet the total number of Turing machines is countable.
    Thus many of the Turing machines in the set $\{ M_i : i\in I\}$ must be identical.
    And thus we can write
    \[\{ M_i : i\in I\}  = \{ M_j' : j\in \N\}.\]
    In other words for every $M_i$ exists an $M_j'$ such that $M_i = M_j'$.
    We partition $\mathcal{O}$ into the sets $\mathcal{O}_j$
    such that for each oracle $O\in \mathcal{O}_j$
    it holds that $M_j'$ establishes that $L(\prop, O) \in \DRPHlevel{k}^O$.
    For all oracles $O$ such that $L(\prop, O) \not \in \DRPHlevel{k}^O$, just put them in an arbitrary class.
    We are now ready for the following calculation.

   \begin{align*}
    \alpha &=  \Pr( \{O \in \mathcal{O} : L(\prop, O) \in \DRPHlevel{k}^{O} \} )  \\
        &=
    \Pr \left(  \{O_i \in \mathcal{O} :  M_i^{O_i} \text{ establishes } L(\prop,O_i) \in \DRPHlevel{k}^{O_i} \} \right)\\
     &= \Pr \left(
     \bigcup_{j\in \N}  \{ O \in \mathcal{O}_j : { M_j'}^{O} \text{ establishes } L(\prop,O) \in \DRPHlevel{k}^{O}   \} 
     \right) \\
     & =  \sum_{j\in \N} \Pr \left(
       \{  O \in \mathcal{O}_j : { M_j'}^{O} \text{ establishes } L(\prop,O) \in \DRPHlevel{k}^{O} : O \in \mathcal{O}_j \} 
     \right) \\
    \end{align*}

    Here the first equation just repeats the definition of $\alpha$.
    The second equation invokes the definition of $M_i$.
    The third equation partitions the set of oracles according to the countable set of Turing machines $M_j'$.
    Finally, the last equation is not an inequality, because the union is actually a \textit{disjoint} union by definition.
    Furthermore, note that last equation holds because the union was \textit{countable} and not uncountable.
    To see that it does not work for uncountable sets, we look at the following small example.
    $1 = \mu([0,1]) = \mu(\cup_{x\in [0,1]} x) \leq \sum_{x\in [0,1]} \mu ({x}) = \sum_{x\in [0,1]} 0 = 0 \neq 1$.

    As each probability is non-negative and $\alpha$ is positive, there must exists at least one real oracle Turing machine $M = M_j'$ such that 
    
    \[ \Pr \left(
       \{ { O\in \mathcal{O} : M}^{O} \text{ establishes } L(\prop,O) \in \DRPHlevel{k}^{O} \} 
     \right) = \beta >0.\]

    Now, what remains is to show that there exists a machine $N$ such that 
    \[\Pr \left(
       \{ { O\in \mathcal{O} : N}^{O} \text{ establishes } L(\prop,O) \in \DRPHlevel{k}^{O} \} 
     \right) = \beta >0.99.\]
    
     Let $\mathcal{O}'$ be the set of oracles that $M$ establishes are in $\DRPHlevel{k}^{O}$.
     In other words,
     \[\mathcal{O}' = \{ O : M^O \text{ establishes } L(\prop,O) \in \DRPHlevel{k}^{O}\}.\]
     We saw that  $\Pr(\mathcal{O}')$ is positive as the probability of choosing an oracle in $\mathcal{O}'$ was positive.
     Now, the Lebesgue density theorem says that for almost every $O \in \mathcal{O}'$, 
     its cylinder density approaches $1$.
     Thus there exists an oracle $O$ and constant $m\in \N$ such that 
     \[\frac{\Pr(\mathcal{O}'\cap [O{\upharpoonright}m])}{\Pr([O{\upharpoonright}m])} > 0.99.\]
    In other words, $M$ will correctly establish 
    if $L(\prop,O) \in \DRPHlevel{k}^{O}$ for a random oracle in $\mathcal{O}'\cap [O{\upharpoonright}m]$ with probability at least $0.99$.
    Note that any oracle $O$ is in $[O{\upharpoonright}m]$ if we just ignore the first $m$ words. 
    Let $n_0$ be the smallest number such that the first $m$ words have all length smaller than $n_0$.
    Then for any oracle $O$ there is another oracle $O' \in [O{\upharpoonright}m]$ such that 
    $1^n \in L(\prop,O)$ if and only if $1^n \in L(\prop,O')$.
    Thus the questions $1^n \in L(\prop,O)?$ will be correctly computed for all $n \geq n_0$ with probability $0.99$ for a random oracle $O$. 
    We modify $M$ in a way that it computes $1^n\in L(\prop,O)$ correctly for every $n<n_0$ 
    and otherwise does exactly what $M$ did before for any larger $n$.
    This is possible as there are only finitely many truth tables that need to be considered.
    We denote the new Turing machine  by $N$.
    Note that $N$ gives the correct answer with probability $0.99$ for every $n$.
    And thus it holds that for a random oracle $O$, $\Pr(N^O \text{ establishes } L(\prop,O) \in \DRPHlevel{k}^{O}) \geq 0.99$.
\end{proof}

    Let us emphasize that \Cref{lem:oneTM} justifies that in subsequent theorems, 
    we may assume that there exists a single real oracle Turing machine $M$, such that $M^O$ establishes $L(\prop, O) \in \DRPHlevel{k}^O$ with positive probability over all oracles $O \in \mathcal{O}$, instead of a set of real oracle Turing machines.

\begin{theorem}
\label{thm:RPHandACanalogue}
Let $\prop:\{0,1\}^* \rightarrow \{0,1\}$ be a function, 
$M$ be a real oracle Turing machine,
and $\mathcal{O}' \subseteq \mathcal{O}$ a set of oracles.
%
We assume 
for all oracles $O\in \mathcal{O}'$ it holds that 
$L(\prop, O)$ is in $\DRPHlevel{k}^O$ and this is established by $M^O$.

Then there exists a family of $k$-depth circuits $(C_n)_{n\in \N}$ of size
$|C_n| \leq 2\cdot\RPHcircuitUpperBound$ (with $\RPHcircuitUpperBound = 2^{2^k\bigO( n^c)}$), for some constant $c$ only depending on $M$,
such that for all oracles $O\in \mathcal{O}'$
it holds that for all $n$ 

 
    \[1^n \in L(\prop,O) \Leftrightarrow C_{n}(a) = 1.\]   
Here,
$a \in \{0,1\}^{\RPHcircuitUpperBound}$
is the input generated in the following way: for all $j \in [\RPHcircuitUpperBound]$,
$a(j)=M^O(1^n, w_j)$,
where $w_j$ is the $j$th element in
$\LLL$
ordered by the lexicographic ordering.
\end{theorem}

Note that the circuit family $C_n$ is independent of the oracle.
Furthermore, recall that $L(\prop,O)$ is a unary language, so we do not need to check for arbitrary $x\in \{0,1\}^*$ if they are a word in the language.
Before we engage in the formal proof, it might be instructive to see an example circuit for the case $k=1$ and $k=2$.

\begin{example}
\label{eg:kis1}
    Assume the notation and conditions as in \Cref{thm:RPHandACanalogue}, with $k=1$ and let $Q_1 = \exists$ be an existential quantifier.
    Using \Cref{lem:ConstrainedQuantifierRanges}, for any $\T$-time Turing machine $M^O$ and any oracle 
    $O$
    we can immediately replace the real set we quantify over, by the sets 
    $\LLL(n,\T,1,1)$.
    For the point of having an example that we can draw in a figure, 
    we assume that $\LLL(n,\T,1,1) = \{w_1,w_2,w_3,w_4\}$ is of size four.
    Thus, we have the following equivalence
    \[1^n\in L(\prop, O) \Leftrightarrow \exists w\in \{w_1,w_2,w_3,w_4\} : M^O(1^n,w) = 1\]
    We generate the binary string $a$, by setting 
    $a(i) = M^O(1^n,w_i)$.
    (We interpret $M(1^n,w_i) \neq 1$ as $M^O(1^n,w_i) = 0$.)
    Using $a$, we see that evaluating the following circuit $C_n$ is equivalent to $1^n\in L(\prop, O)$.
    Note that the input values to the circuit, the results for $M^O(\dots)$ may differ depending on the oracle $O$.
    \begin{center}
        \includegraphics[page = 2]{figures/TreeRPH3.pdf}    
    \end{center}
\end{example}
The previous example was straightforward and should emphasize the simplicity of our result, despite a bit of formalism.
As the case $k=1$ is maybe not representative of the general case, we provide another example with $k=2$.

\begin{example}
    \label{ex:depth2tree}
    Assume the notation and conditions as in \Cref{thm:RPHandACanalogue}, with $k=2$ and let $Q_1 = \exists$ be an existential quantifier and $Q_2 = \forall$ be a universal quantifier.
    Using \Cref{lem:ConstrainedQuantifierRanges}, for any oracle $O$ and any $\T$-time Turing machine $M^O$, we can immediately replace the real set we quantify over, by the sets 
    $\LLL(n,\T,1,2)$ and $\LLL(n,\T,2,2)$.
    For the point of having an example that we can draw in a figure, 
    we assume that 
    $\LLL(n,\T,1,2) = \{v_1,v_2,v_3\}$
    and
    $\LLL(n,\T,1,2) = \{w_1,w_2,w_3\}$
    are both of size three.
    Thus, we have the following equivalence:
    \[1^n\in L(\prop, O) \Leftrightarrow \exists v\in \{v_1,v_2,v_3\} \ \forall w \in \{w_1,w_2,w_3\}  : M^O(1^n,w) = 1\]
    We generate the binary string $a$ of length $9$.
    However, for ease of notation, we write 
    $a(i,j)$ instead of $a(3(i-1)+j)$.
    We set $a(i,j) = M(1^n,v_i,w_j)$.
    (We interpret $M(1^n,w_i) \neq 1$ as $M(1^n,w_i) = 0$.)
    Using $a$, we see that evaluating the following circuit $C_n$ is equivalent to $x\in L(\prop, O)$.
    Again, note that $a$ depends on the oracle $O$.
    \begin{center}
        \includegraphics[page = 4]{figures/TreeRPH3.pdf}    
    \end{center}
    To see the equivalence, we can think of evaluating the circuit from top to bottom.
    Say the existential quantifier $Q_1$ chooses $v = v_2$,
    then the top node also  chooses $v_2$ and all leaf nodes below have $v = v_2$.
\end{example}

\bigskip

\begin{proof}
Let $n\in \N$. By using \Cref{lem:ConstrainedQuantifierRanges}, for any oracle $O$ and any $\T$-time Turing machine $M^O$, we can immediately reduce deciding
\begin{equation}
\label{eq:ActualMachine}
    Q_1 u_1 \in \R^{q(n)} \cdots Q_k u_k \in \R^{q(n)} : M^O(1^n, u_1,\ldots,u_k)=1
\end{equation}
to deciding
\begin{equation}
\label{eq:FunctionallyEquivalentMachine}
    Q_1 u_1 \in \LLL(n,\T, 1, k) \cdots Q_{k} u_{k} \in \LLL(n, \T, k, k) : M^O(1^n,u_1,\ldots,u_k)=1.
\end{equation}
    
\noindent Following the convention of Boolean circuits in \Cref{def:BooleanCircuits}, we construct a directed acyclic graph where the non-input vertices are referred to as gates and each gate will be labeled by either $\lor$ or $\land$. It helps to think of an constant-depth circuit to have a single gate, either $\land$ or $\lor$, at the top and then $(k-1)$ layers of gates of alternating $\land$ and $\lor$ gates feeding into their previous layers.
See \Cref{ex:depth2tree} for an illustration.
Now, we give the construction of an appropriate $k$-depth circuit that simulates a machine $M^O$ that establishes $L(\prop, O) \in \DRPHlevel{k}^O$ for all oracles $O \in \mathcal{O}'$.

\paragraph{Circuit Construction from the Top.}

\begin{enumerate}
    \item \textbf{Root node:} The top node is $\lor$ in case that $Q_1 = \exists$ and $\land$ in case that $Q_1 = \forall$.
    The number of children equals $|\LLL(n,\T,1,k)|$.
    Furthermore, we label every edge with the corresponding $w_j(1) \in \LLL(n,\T,1,k)$.
    We say that the top layer is the first layer.
    Note that we constructed already all the nodes from the $2$nd layer.
    \item \textbf{Intermediate nodes:} All nodes on the $i$th layer are  $\lor$ in case that $Q_i = \exists$ and $\land$ in case that $Q_i = \forall$.
    The number of children equals $\LLL(n,\T,i,k)$.
    Furthermore, we label every edge with the corresponding $w_j(i) \in \LLL(n,\T,i,k)$.
    \item \textbf{Input nodes:}
    The input nodes are on the last layer $k+1$ and are already constructed.
    For each input node, we look at the unique path from that node to the root. 
    Say that the labels $w_j = (w_{j}(1),\ldots,w_j(k))$ were on the path. 
    Then the value at the input node equals $a(j) = M^O(1^n,w_j)$.
\end{enumerate}

We want to point out that as each node on level $i$ has $|\LLL(n,\T,i,k)|$ children, the total number of children equals 
\[S = |\LLL(n,\T,1,k)| \times \ldots \times |\LLL(n,\T,k,k)|.\]

To see the equivalence of evaluating $C_n(a)$ with $x\in L(\prop, O)$,
we note that evaluating it from top to bottom 
corresponds to evaluating the \Cref{eq:FunctionallyEquivalentMachine} from left to right.
Now, for each oracle for which $M^O$ establishes that $L(\prop, O) \in \DRPHlevel{k}^O$, we find that $C_n$ (for all $n$) must, for any such oracle $O$, correctly compute $\prop(\truthtable(O, n))$ when its input nodes take the values of $a(j)=M^O(1^n, w_j)$ where $w_j$ is the $j$th element in $\LLL$ ordered by the lexicographic ordering.
\end{proof}

Now, the issue is that the leaf nodes of the circuit use $M^O$ as their input, which can be unpredictable when $M^{O}$ may make at most $\T(n)$ many queries to the oracle $O$. 
Using arguments similar to ones used in \cite{FSS84, Cai86}, by adding one additional layer of quantifier (either $\exists$ or $\forall$ depending on whether $k$ is even or odd and whether $L(\prop, O) \in \SRPHlevel{k}$ or $L(\prop, O) \in \PRPHlevel{k}$), one can replace the machine $M^{O}$ with another machine $M'^{O}$ that makes at most one query to the oracle $O$. As this can be done for every circuit in $(C_n)_{n \in \mathbb{N}}$, the result of this procedure is a family of $k+1$-depth circuit with its input nodes being the output of $M'^{O}$. We state the argument here for completeness.

The following definition will be useful for our proof of this theorem. The notion of decision trees has been extensively studied in the context of black-box query complexities but we use the idea to simplify the presentation of our proof; the survey article \cite{BW02} serves as an excellent resource. Towards that, we use the definition of decision trees from \cite{CMP22}.

\begin{definition}[Decision Trees]
\label{def:decisionTrees}
A decision tree computing a function $f : \{0,1\}^m \rightarrow \{0,1\}$ is a binary tree with leaf nodes having labels in $\{0,1\}$, each internal node is labeled by a variable $x_i$ (for an $i \in [m]$) and has two outgoing edges, labeled $0$ and $1$. On input $x \in \{0,1\}^m$, the tree’s computation proceeds by computing $x_i$ as indicated by the node’s label and following the edge indicated by the value of the computed variable. The output value at the leaf reached by proceeding the computation on $x$, let us denote by $\mathcal{T}(x)$, must be such that $f(x)=\mathcal{T}(x)$. 

Given a function  $f : \{0,1\}^m \rightarrow \{0,1\}$ and a decision tree $\mathcal{T}$ computing it, define the function $\leaf$ that takes an input $x \in \{0,1\}^m$ and outputs the leaf of $\mathcal{T}$ reached on input $x$. For an input $x \in \{0,1\}^m$, let $\route(x)$ denote the unique path in $\mathcal{T}$ from its root to $\leaf(x)$. We use $\mathcal{P}(\mathcal{T})$ to denote the set of all paths from the root to a leaf in $\mathcal{T}$.

We assume that the decision trees computing functions $f$ contain no extraneous leaves, i.e., for all leaves, there is an input $x \in \{0,1\}^m$ that reaches that leaf. We also assume that for every path $R$ in $\mathcal{P}(\mathcal{T})$ and index $i \in [m]$, the variable $x_i$ is queried at most once on $R$.
\end{definition}

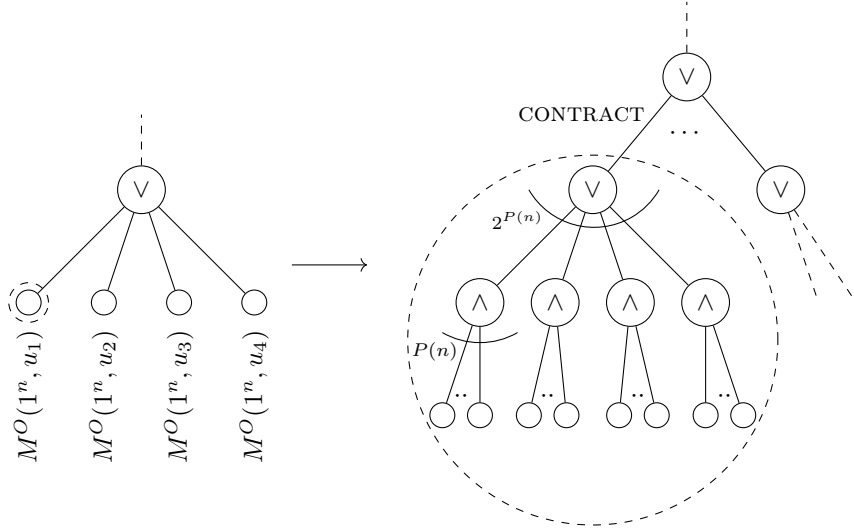
\begin{figure}
    \centering
    \begin{tikzpicture}

        \node[circle,draw] (parent) at (-3,0) {$\lor$};
        \node[circle,draw] (child1) at (-4.5,-1.5) {};
        \node[circle,draw] (child2) at (-3.5,-1.5) {};
        \node[circle,draw] (child3) at (-2.5,-1.5) {};
        \node[circle,draw] (child4) at (-1.5,-1.5) {};

        \node[rotate=90] (child1-label) at (-4.5, -2.75) {$M^{O}(1^n, u_1)$};
        \node[rotate=90] (child2-label) at (-3.5, -2.75) {$M^{O}(1^n, u_2)$};
        \node[rotate=90] (child3-label) at (-2.5, -2.75) {$M^{O}(1^n, u_3)$};
        \node[rotate=90] (child4-label) at (-1.5, -2.75) {$M^{O}(1^n, u_4)$};

        \draw [-] (parent) -- (child1);
        \draw [-] (parent) -- (child2);
        \draw [-] (parent) -- (child3);
        \draw [-] (parent) -- (child4);
        \draw [dashed] (parent) -- (-3,1);

        \node[draw,circle,dashed,minimum size=15pt] (OldCircuit) at (-4.5,-1.5) {};

        \draw [->] (-1,-1) -- (0,-1);

        \node [circle,draw] (grandparent) at (4.25,1.5) {$\lor$};

        \node [circle, draw] (absenteeparent) at (5.5,0) {$\lor$};
        
        \node[circle,draw] (parent) at (3,0) {$\lor$};
        \node[circle,draw] (child1) at (1.5,-1.5) {$\land$};
        \node[circle,draw] (child2) at (2.5,-1.5) {$\land$};
        \node[circle,draw] (child3) at (3.5,-1.5) {$\land$};
        \node[circle,draw] (child4) at (4.5,-1.5) {$\land$};

        \node[circle,draw] (grandchild11) at (1,-3) {};
        \node (grandchild1-dots) at (1.25,-2.75) {$\cdot \cdot$};
        \node[circle,draw] (grandchild12) at (1.5,-3) {};

        \node[circle,draw] (grandchild21) at (2.15,-3) {};
        \node (grandchild2-dots) at (2.4,-2.75) {$\cdot \cdot$};
        \node[circle,draw] (grandchild22) at (2.65,-3) {};

        \node[circle,draw] (grandchild31) at (3.35,-3) {};
        \node (grandchild2-dots) at (3.6,-2.75) {$\cdot \cdot$};
        \node[circle,draw] (grandchild32) at (3.85,-3) {};

        \node[circle,draw] (grandchild41) at (4.5,-3) {};
        \node (grandchild4-dots) at (4.75,-2.75) {$\cdot \cdot$};
        \node[circle,draw] (grandchild42) at (5,-3) {};
       
        \node[draw,circle,dashed,minimum size=140pt] (DNFpartInNewCircuit) at (3,-2) {};
    
        \draw [-] (grandparent) -- (parent);
        \draw [-] (grandparent) -- (absenteeparent);
        \draw [-] (parent) -- (child1);
        \draw [-] (parent) -- (child2);
        \draw [-] (parent) -- (child3);
        \draw [-] (parent) -- (child4);
        \draw [-] (child1) -- (grandchild11);
        \draw [-] (child1) -- (grandchild12);
        \draw [-] (child2) -- (grandchild21);
        \draw [-] (child2) -- (grandchild22);
        \draw [-] (child3) -- (grandchild31);
        \draw [-] (child3) -- (grandchild32);
        \draw [-] (child4) -- (grandchild41);
        \draw [-] (child4) -- (grandchild42);

        \draw[black] (2.15,0) arc (-150:-30:1);
        \node (exponential) at (2,-0.4) {\scriptsize $2^{T(n)}$};

        \draw[black] (1,-1.9) arc (-120:-60:1);
        \node (poly) at (0.9,-2.15) {\scriptsize $T(n)$};

        \draw [dashed] (absenteeparent) -- (6,-1.5);
        \draw [dashed] (absenteeparent) -- (6.5,-1.5);
        \node [rotate=0] (dots) at (4.25,0.75) {$\cdots$};
        \draw [dashed] (grandparent) -- (4.25,2.5);
        \node [rotate=0] (Contract) at (2.85,1) {$\textsc{contract}$};

    \end{tikzpicture}
    \caption{Illustration of the circuit construction mentioned in the proof of \Cref{thm:CircuitLowerBoundsImplyNoMembershipInRPH}. In the first step of the construction, each input node in the old circuit gets replaced by depth-$2$ circuits which have a fan-in of at most $\T(n)$ for the gates that are the parents of the input nodes and have a fan-in of at most $2^{\T(n)}$ for the gates that are the grandparents of the input nodes; here $\T(n)$ is a polynomial. Furthermore, its input nodes are outputs of $M'^{O}$. The second step of the construction is to contract consecutive edges with the same labels. The final depth of the new circuit is $1$ more than the depth of the original circuit.}
    \label{fig:MachineMakesOneQuery}
\end{figure}

\begin{theorem}
    \label{thm:CircuitLowerBoundsImplyNoMembershipInRPH}
Let $\prop:\{0,1\}^* \rightarrow \{0,1\}$ be a function, 
$M$ be a real oracle Turing machine
and \(\mathcal{O}' \subseteq \mathcal{O}\) a set of oracles.
We assume for all oracles \(O \in \mathcal{O}'\) it holds that
$L(\prop, O)$ is in $\DRPHlevel{k}^O$ and this is established by $M^O$.

Then there exists a family of $(k+1)$-depth circuits $(C_n)_{n\in \N}$ of size
$|C_n| \leq {(2^n)}^{\bigO(\polylog(2^n))}$, 
such that for all
oracles  $O \in \mathcal{O}'$ it holds that for all $n$
%

    \[1^n \in L(\prop, O) \Leftrightarrow C_{n}(a) = 1.\]
Here, 
$a = \truthtable(O,n)$.


\end{theorem}
Recall that $\truthtable(O,n)$ is the truth table of $O$ and tells us which words of size $n$ are members of the oracle $O$.
\begin{proof}
Let $(D_n)_{n \in \N}$ of size $|D_n| \leq 2 \cdot S$, with $S\leq (2^n \cdot \multiplier^{P(k \cdot q(n)+n)})^{2^k}$, be the family of $k$-depth circuits we get by invoking \Cref{thm:RPHandACanalogue} for $L(\prop, O)$.
For any fixed \(n\),
we will now construct a circuit $C_n$, using the following observations.
\begin{enumerate}
\item Depending on whether $k$ is even or odd and depending on whether $L$ is in $\SRPHlevel{k}$ or $\PRPHlevel{k}$ the intermediate nodes right above the input nodes of $D_n$ will be either labeled by $\lor$ or by $\land$. This information is important as our construction $C_n$ will depend on this label and we call it the \emph{last} label of $D_n$.
\item The $i^\text{th}$ input to the circuit $D_n$ corresponds to the output of the Turing machine $M^O(1^n, u)$ where $u \coloneqq (u_1, \ldots, u_k)$ is the $i^\text{th}$ element of the set $\mathcal{L}$ as stated in \Cref{thm:RPHandACanalogue}. Furthermore, for every $u \in \mathcal{L}$, the machine $M^{O}$ makes at most $\T(n)$ many \emph{adaptive} queries to oracle $O$.
\end{enumerate}
Let $\mathcal{T}_u$ be the decision tree (\Cref{def:decisionTrees}) capturing the computation of $M^{O}$ on $1^n$ and $u$ that proceeds by computing queries to $\truthtable(O, t)$ which is the string provided by $O$.
This additionally means that the output $M^O(1^n, u)$, with $z$ being the string provided by $O$, is consistent with $\mathcal{T}_u(z)$. 
Let $\DNF(\mathcal{T}_u)$ denote a $\DNF$ formula (\Cref{def:CNFnDNFformulas}) that accepts all the paths in $\mathcal{T}_u$ reaching leaves of output value $1$; w.l.o.g., we can assume such a path exists otherwise we can replace the node $M^O(1^n, u)$ with $0$. Alternatively, let $\DNF(\neg \mathcal{T}_u)$ denote a $\DNF$ formula that accepts all the paths in $\mathcal{T}_u$ reaching leaves of output value $0$. Using De Morgan's law we can see that the $\neg \DNF(\neg \mathcal{T}_u)$ will be a $\CNF$ formula (\Cref{def:CNFnDNFformulas}) that accepts all the paths in $\mathcal{T}_u$ reaching $1$; let us denote this by $\CNF(\mathcal{T}_u)$. We will now construct $C_n$ from $D_n$ in the following way. If the \emph{last} label of $D_n$ was $\lor$ ($\land$) then for every $i^{\text{th}}$ input node of $D_n$ do the following.
\begin{enumerate}
    \item Replace the input $M^O(1^n, u)$ (where $u$ is the $i^\text{th}$ element in $\mathcal{L}$ under a fixed ordering) with a circuit constructed from the $\DNF(\mathcal{T}_u)$ ($\CNF(\mathcal{T}_u)$) with its output feeding into the $\lor$ ($\land$) as seen in \Cref{fig:MachineMakesOneQuery}. The input node of the resulting circuit can now be interpreted as the output of $M'^{O}$ on $1^n$ and $u$ and more importantly it makes a single query to the string provided by $O$. Furthermore, the resulting $(k+2)$-depth circuit contains two consecutive layers with the same label.
    \item Contract the edges between the nodes of the consecutive layers of the same labels. The resulting circuit, which we denote by $C_n$, is of depth $k+1$ and of size at most $2^{\T(n)} \cdot \T(n)$ times the size of $D_n$.
\end{enumerate}

As the inputs to $C_n$ makes at most one query to oracle $O$, the input to $C_n$ (which computes property $\prop$) is of length $2^n$.
Now, we conclude that the input nodes are queries to the oracle $O$ and see its size $|C_n| \leq 2 \cdot 2^{\T(n)} \cdot \T(n) \cdot S$.
The circuit bounds in terms of this input length $N=2^n$ is at most $N^{\bigO(\polylog N)}$ and the circuits $(C_n)_{n \in \mathbb{N}}$ compute $\prop$; the big $\bigO$ notation is in terms of $N$. Hence, proving the statement of \Cref{thm:CircuitLowerBoundsImplyNoMembershipInRPH}.
\end{proof}

The following remark will be useful for \Cref{thm:LevelSep} in \Cref{sec:SepLevelsRPH}.
\begin{remark}
\label{rem:FaninAtmostPoly} Each circuit $C_n$ in the family of $(k+1)$-depth circuits $\{C_n\}_{n \in \mathbb{N}}$ resulting from \Cref{thm:CircuitLowerBoundsImplyNoMembershipInRPH} has the property that all the intermediate nodes right above the input nodes of $C_n$ have a fan-in of at most $\poly(n)$; also illustrated in \Cref{fig:MachineMakesOneQuery}.  
\end{remark}

\section{Separations of \RPH}
\label{sec:RPHSeparations}

In this section, using our results from \Cref{sec:ACanalogueForRPH}, we will present oracle separations between $\RPH$ and well-studied complexity classes $\PSPACE$ and $\BQP$. We also present oracle separations within the layers of $\RPH$ itself.

\subsection{Interpreting the Circuit Bounds from \Cref{sub:CircuitUpperBounds}}
\label{sub:InterpretationCircuitUpperBounds}

In \Cref{sec:ACanalogueForRPH} (more precisely in \Cref{thm:RPHandACanalogue} and \Cref{thm:CircuitLowerBoundsImplyNoMembershipInRPH}) we were able to show how an $\RPH^O$ computation can be translated to a family of circuits of constant depth. 
In this section, we give an intuitive overview of how we can use the family of circuits to find oracle separations.

Suppose we are given a machine $M^O$ that establishes $L(\prop, O) \in \DRPHlevel{k}^O$ with probability $\alpha > 0$ over all random oracles $O$.
The input of machine $M^O$ is a unary string $1^n$.
If $L(\prop, O)$ is established to be in $\DRPHlevel{k}^O$ by $M^O$, then
it accepts all $n$ for which $\prop(\truthtable(O,n)) = 1$.

Applying \Cref{thm:RPHandACanalogue} and \Cref{thm:CircuitLowerBoundsImplyNoMembershipInRPH} on $M^O$ for any $n$ yields a ``small'' constant-depth circuit $C_n$ with (at most) $2^n$ inputs.
Here, the input is not the same as the input of the machine $M^O$.
Instead, the leaf nodes of the computation tree structure of $M^O$ are the input nodes of $C_n$.
In \Cref{thm:CircuitLowerBoundsImplyNoMembershipInRPH}, it was shown that each computational path only needs to query the oracle once.
As such, the leaf nodes will take the values of $\truthtable(O, n)$.
Of course, for a single oracle, these values are constant.
However, we require that $M^O$ establishes that $L(\prop, O) \in \DRPHlevel{k}^O$ for a random oracle with probability $\alpha > 0$.
For each such oracle $O$, the circuit $C_n$ (for any $n$) must correctly compute $\prop$ for $\truthtable(O, n)$.
Using this requirement, we provide various oracle separations, as circuit lower bounds show for various functions that ``small'' circuits cannot successfully compute $\prop$.

We show two types of separations.
In \Cref{sec:RPHvsPSPACE} and \Cref{sec:SepLevelsRPH}, we provide results relative to a random oracle.
Here, we use \Cref{lem:oneTM} to argue that if, for some probability $\alpha > 0$, $L(\prop, O) \in \DRPHlevel{k}^O$, then there also exists some $\beta > 0$ such that some machine $M^O$ establishes this with probability $\beta$ for a random oracle $O$.
Then, we argue using \Cref{thm:RPHcircuitSizeAndNoGoTheorem} that for the set oracles for which $M^O$ establishes $L(\prop, O) \in \DRPHlevel{k}^O$, we find the ``small'' circuit.

Similarly, in \Cref{sec:BQPvsRPH}, we separate complexity classes for some oracle $O$.
In this case, we argue that the machine must establish $L(\prop, O) \in \DRPHlevel{k}^O$.
Again, this would result in a ``small" circuit that would be able to compute $\prop(\truthtable(O, n))$ for all truth tables, which we will use to get a contradiction and show that some oracle must separate the classes.


\subsection{Separating $\RPH$ and $\PSPACE$}
\label{sec:RPHvsPSPACE}
In this section, we prove \Cref{thm:PSPACEnotRPH}, which we restate for the benefit of the reader.

\PSPACEnotRPH* 

Assume $O$ is a random oracle.
That is, for each string $x\in \{0,1\}^*$, we flip a fair coin to decide if $x\in O$ or not.
We show $\PSPACE^O \not\subset \RPH^O$ with probability one, which separates $\RPH^O$ and $\PSPACE^O$.
Furthermore, we show that, $\RPH^O \subseteq \PSPACE^O$ for all $O$, implying that $\RPH^O \subsetneq \PSPACE^O$ with probability one.

First, we work towards $\PSPACE^O \not\subset \RPH^O$ with probability one.
For this, we show that $\RPH^O$ does not contain a problem called $\ParityA$, but $\PSPACE^O$ does.
This closely follows Cai's proof that $\PSPACE^O\not\subset\PH^O$ with probability one \cite{Cai86}.

\begin{definition}[$\Parityn$]
    The parity function over $n$ boolean variables is given by \emph{$\Parityn(x_1, x_2, \dots, x_n)$} = $x_1 \oplus x_2 \oplus \dots \oplus x_n$, where $\oplus$ is the exclusive or operator.
\end{definition}

\begin{definition}[$\ParityA$ \cite{FSS84}]
Given a language $O \subseteq \{0,1\}^*$, 
\begin{equation*}
    \ParityA = \{ 1^n \mid \text{\emph{the number of binary strings of length $n$ in $O$ is odd}}\}.
\end{equation*}
\end{definition}

Note that whether there is a subscript defines which version of $\Parity$ we are talking about:
if there is no subscript, we talk about the parity function,
if the subscript is a language or its eponymous oracle, we talk about the language $\ParityA$.
We note that there are $N=2^n$ binary strings of length $n$.
So, solving $\ParityA$ on input $n$ requires us to solve $\Parityn$ on $N$ variables.

$\ParityA$ is an example of an \textit{oracle property} \cite{Ang980};
we ask our machine to answer a question about a property of the language $O$.

First, we show that $\ParityA$ is indeed in $\PSPACE^O$.

\begin{lemma}[{\cite{FSS84}}]
    For any language $O$, $\ParityA \in \PSPACE^O$.
\end{lemma}
\begin{proof}
    A $\PSPACE$-machine can simply iterate over all possible strings of length $n$.
    It uses a single bit $B$ to track the parity of $O$.
    We initiate $B$ with the value $0$.
    For each string $s$, it sets $B$ to $B \oplus (s \in O)$.
    After the $\PSPACE$-machine has iterated over all possible strings $s$ of length $n$, it will have calculated $\ParityA$ in polynomial space.
\end{proof}
Now, we formally separate $\RPH^O$ and $\PSPACE^O$ by showing that $\ParityA$ is not in $\RPH^O$.
We use \Cref{thm:CircuitLowerBoundsImplyNoMembershipInRPH}
along with the following Lemma from Cai to show that $\RPH^O$ can not contain $\ParityA$.

\begin{lemma}[{\cite[Theorem~2.1]{Cai86}} \label{thm:cai}]
    For all $k\geq 2$, there is a sequence $\alpha_n$, where $\alpha_n \rightarrow \frac{1}{2}$ as $n\rightarrow \infty$, such that all depth $k$ Boolean circuits with $2^n$ input bits and size bounded by $$2^{(2^n)^{\frac{1}{4(k+1)}}}$$ err on more than $\alpha_n\cdot 2^{(2^n)}$ of the $2^{2^n}$ many inputs when computing $Parity$ on $2^n$ input bits.
\end{lemma}

\begin{lemma}
    For a random oracle $O$, $\ParityA \not\in \RPH^O$ with probability 1.
\end{lemma}

\begin{proof}
    
    Suppose $\ParityA \in \RPH^O$ with probability $\alpha > 0$.
    Then, there must exist a single real oracle Turing machine $M$ such that $M$ establishes $\ParityA \in \RPH^O$ for a random oracle $O$ with probability $\beta > 0$ by \Cref{lem:oneTM}.

    Using this real oracle Turing machine~$M$ for the property $\prop = \ParityA$, we invoke \Cref{thm:CircuitLowerBoundsImplyNoMembershipInRPH} and get a family of circuits $(C_n)_{n\in\N}$ such that the following holds.
    We denote with $n$ the input size of the machine $M$ and denote $N= 2^n$ for the number of the number of leaf nodes of the corresponding Boolean circuit $C_n$ as per \Cref{thm:CircuitLowerBoundsImplyNoMembershipInRPH}.
    Recall that $M$ decides if a word $w=1^n$ is in $L(\ParityA)$ and $C_n$ does the same. 
    But we assume that $C_n$ receives the truth table $\truthtable(O,n)$ of length $N = 2^n$.
    Remember that $C_n$ has (constant) depth $k$ and size at most $N^{\bigO(\polylog(N))}$. 
    As $M$ does this correctly for a $\beta > 0.99$ fraction of the oracles by \Cref{lem:oneTM}, so does $C_n$, by \Cref{thm:CircuitLowerBoundsImplyNoMembershipInRPH}.
    As such, if $M^O$ establishes $\ParityA \in \RPH^O$ with probability $\beta > 0.99$, the circuit should compute the correct answer with at least probability $0.99 > \frac{1}{2}$.)
    However, by \Cref{thm:cai}, such a circuit should have size \(2^{\Omega(N^\frac{1}{k})}\).
    As \(2^{\Omega(N^\frac{1}{k})}\) grows asymptotically faster than $N^{\bigO(\polylog(N))}=2^{\bigO(\polylog(N))}$,
    for large enough $n$, we have \(2^{\Omega(N^\frac{1}{k})} \geq 2^{\bigO(\polylog(N))}\) for $N=2^n$.
    As such, we reach a contradiction, and we prove $\ParityA \not\in \RPH^O$ with probability 1.
\end{proof}

We observe that $\ParityA$ is not in $\RPH^O$, but is in $\PSPACE^O$.
As such, we conclude that $\PSPACE^O \nsubseteq \RPH^O$ for a random oracle $O$.
Note that for any given circuit which attempts to compute \Parity, it can succeed for some specific oracles, which is why we cannot claim $\RPH^O \nsubseteq \PSPACE^O$ for all $O$.

Now, we shift our attention to the inclusion of $\RPH^O \subseteq \PSPACE^O$ given any binary oracle $O$.
Note that this is not trivially following from 
$\RPH \subseteq \PSPACE$, as \RPH can guess real numbers; such guesses are a priori difficult to simulate using polynomial space.

\begin{lemma}
    \label{lem:RPHSubsetPSPACE}
    $\RPH^O\subseteq\PSPACE^O$ for every binary oracle $O$.
\end{lemma}

As a preparation, we need the following \Cref{lem:PSPACESelfLow},
which seems to be folklore.
However, to the best of the authors' knowledge, there is no peer-reviewed source to justify the claim. 
We have only found a blog post entry by Aaronson~\cite{aaronson2014lens}.
For completeness, we present a simple proof.
Recall that a complexity $C$ class is \textit{self-low}, if $C^C = C$.
\begin{lemma}
    \label{lem:PSPACESelfLow}
    $\PSPACE$ is self-low, assuming inputs of oracle queries are at most of polynomial length.
\end{lemma}

\begin{proof}
    When we query some $\PSPACE$-oracle as a subroutine, the oracle's input must have some size $\poly(n)$.
    Thus, using the $\PSPACE$-machine to solve this oracle query would use at most $\poly(\poly(n))=\poly(n)$ space.
    In conclusion, any $\PSPACE$-oracle used within a $\PSPACE$-machine can be replaced by a machine that takes at most $\poly(n)$ space.
    Note that we assume that oracle queries are at most of polynomial length, to ensure that the machine's oracle tape is bound to the same polynomial-space requirement as its work tape.
\end{proof}

We are now ready to show that $\RPH^O$  is contained in $\PSPACE^O$, for binary oracles.

\begin{figure}[t]
\centering
\includegraphics{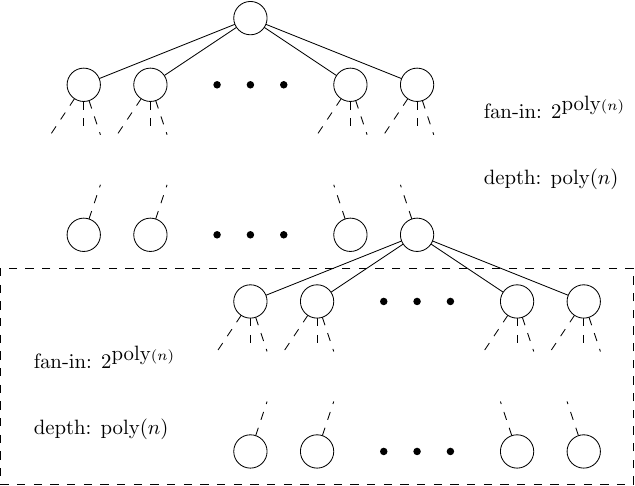}

\label{fig:my_label}
\caption{
A sketch of the circuit a $\PSPACE$-machine uses to solve $\RPH^O$.
The top levels of the circuit represent the different variable assignments.
Each level represents a variable, with each variable getting one of $2^{\poly(n)}$ possible assignments.
Each variable assignment is encoded by a polynomial-sized $\LogicRPH$ formula.
As there are at most $\poly(n)$ variables, the depth is $\poly(n)$.
At the bottom layers, all variables have been assigned a value.
These then handle the different possible oracle queries the machine may attempt to ask.
There are at most $2^{\poly(n)}$ possible different oracle queries of polynomial length the machine might ask, so the fan-in is at most $2^{\poly(n)}$.
Furthermore, the depth is $\poly(n)$, as the machine may ask at most $\poly(n)$ different queries, as it can ask at most one per time step.
}
\end{figure}

\begin{proof}[Proof of \Cref{lem:RPHSubsetPSPACE}]
    Let $L$ be a language in $\DRPHlevel{k}^O$.
    Consider the corresponding circuits \( (C_n)_{n\in \mathbb{N}} \) with inputs \(a\) (depending on \(L\) and \(x \in \{0,1^*\}\)) which we obtain from an application of \Cref{thm:RPHandACanalogue} to \(L\).
    We will establish how the computations of an arbitrary such \(C_n\) on an arbitrary such input \(a\) can be carried out in polynomial space.
    This is obvious as soon as we are able to compute \(a\) in polynomial space, because the computation of the conjunction/disjunction of polynomial-space computable Boolean values can be updated one Boolean value at a time, and hence still using polynomial space.

    To compute the \(j\)th bit of \(a\) we need to evaluate \(M^O(x,w_j)\) for \(x \in \{0,1\}^*\) and \(w_j \in \LLL\) as specified by \Cref{thm:RPHandACanalogue}.
    For this, we need to overcome two obstacles which a priori prevent us from directly using the fact that \(\RPH \subseteq \PSPACE\)~\cite{Can88, R92B}: (1) \(w_j \in \LLL\) is not encoded in binary and (2) the access of \(M^O\) to \(O\).
    
    We first address (1) and formulate \(M^O(x,w_j)\) through a computation of a polynomial-time real Turing machine with oracle access on binary input.
    This reformulation can be carried out in polynomial space.
    Since each number in $\LLL$ is generated by finding an assignment of variables such that a polynomial-time real Turing machine either evaluates to true or false (see the proof of \Cref{lem:ConstrainedQuantifierRanges}), we can use the set of real Turing machines to encode~$\LLL$.
    Indeed, there may be multiple sets $\LLL$, but the exact set we obtain is not important, as long as it maintains the property from \Cref{lem:ConstrainedQuantifierRanges}.
    Since these real Turing machines are discrete structures, they can be encoded by binary strings and thus a $\PSPACE^O$-machine can generate the first real Turing machine in lexicographical order or the next real Turing machine that is representable by a binary string of length at most polynomial in the input in polynomial space, given the previous real Turing machine.
    As such, we are able to iterate over all polynomial-time real Turing machines in polynomial space.
    Now, we can determine our variables to be equal to the resulting variables generated from the appropriate real Turing machine, rather than explicitly setting them.
    As each variable is replaced by a polynomial size Turing machine, the size of the program will still be polynomial.

    To address (2) we describe a $\PSPACE^O$-machine that simulates a polynomial-time real Turing machine with oracle access by querying an $\RPH$-machine.
    From \Cref{lem:PSPACESelfLow}, we know $\PSPACE$ is self-low (if we do not allow exponentially long oracle queries), which means by definition $(\PSPACE)^{(\PSPACE)}=\PSPACE$.
    As $\RPH \subseteq \PSPACE \subseteq \PSPACE^O$~\cite{Can88}, we can use an oracle to $\RPH$ as part of our $\PSPACE^O$-machine.
    So, we would like to simply ask the $\RPH$ oracle to carry out the computation required for the computation of the input bits.
    However, one issue remains: the $\RPH$ oracle does not have access to an oracle to $O$.
    To circumvent this, for each input bit, the $\PSPACE^O$-machine generates all combinations of a polynomial amount of queries to $O$. Each $\RPH$-machine will query $O$ at most a polynomial amount, as it only evaluates a quantifier-free formula, as it guesses the $\RPH$-sentences corresponding to the values in $\LLL$ before asking the $\RPH$-oracle to evaluate the result.
    Then, at least one of these combinations must contain all queries that the $\RPH$-machine makes.
    The $\RPH$-machine then takes as part of its input a list of oracle queries and their results.
    Then, the $\RPH$-machine instead refers to its input registers where the answers are stored.
    Furthermore, if it ever encounters an oracle query that is not given as input, the $\RPH$-machine will immediately write a $0$ to its output register and set its state to $q_\texttt{halt}$.

    We note that the path of the $\RPH$-machine will always be the same, provided the oracle queries are passed as input.
    So, any set of oracle queries passed by the $\PSPACE^O$-machine will either result in the $\RPH$-machine halting prematurely or finding the correct solution.
    If, for any of the combinations of oracle queries, the $\RPH$ machine returns true, the input bit is set to \(1\).
    Otherwise, the input bit is set to \(0\).

 \end{proof}

As $\RPH^O \subseteq \PSPACE^O$ for every binary oracle $O$ and $\PSPACE^O \not\subset \RPH^O$ for a random oracle $O$ with probability 1, we conclude that $\RPH^O \subsetneq \PSPACE^O$ for a random oracle with probability 1 and prove \Cref{thm:PSPACEnotRPH}.

\begin{remark}
It is curious to note that $\RPH^O$ is not necessarily contained in $\PSPACE^O$, if we consider \textit{real} oracles.
Specifically, Meer and Wurm~\cite{MeerWurm2025} showed that $\NP^\Z \subsetneq \ER^{\Z}$ and 
$\ER^\Z \not\supseteq \PSPACE^\Z$.
The underlying reason is that $\ER^\Z$ can solve undecidable problems such as Diophantine equations, whereas a binary Turing machine cannot make efficient use of real oracles.
\end{remark}

\subsection{Separating the Levels of \(\RPH\) (and \(\PH\))}
\label{sec:SepLevelsRPH}
In this section, we prove \Cref{thm:RPHlayers}, which we repeat here for the convenience of the reader.

\RPHlayers*


To prove \Cref{thm:RPHlayers}, we show that there exists a function that $\SRPHlevel{k+1}^O$ can compute, yet $\DRPHlevel{k}^O$ cannot compute.
To do this, we use the following result from Rossman et. al:

\begin{lemma}[{\cite[Average-Case Depth Hierarchy Theorem]{Rossman2015FOCS}}]
    \label{thm:averageCaseDepthHierarchyTheorem}
    For all constants $k \in \N$, there exists a family $\{ \sipser \}_{N \in \N}$ of Boolean functions such that
    \begin{enumerate}
        \item $\sipser$ is an $N$-variable Boolean function computable in depth-$(k+1)$ \AC, and yet
        \item Any depth-$(k+1)$ circuit $C$ of $\quasipoly{N}$ size and bottom fan-in $\polylog(N)$ agrees with $\sipser$ on at most a $\frac{1}{2}+N^{-\Omega(1/(k+1))}$-fraction of inputs.
        That is, the fan-in of the bottom-level gates is at most $\polylog(N)$.
    \end{enumerate}
\end{lemma}

We now define an oracle property that calculates $F$ for some language $O$.
\begin{definition}[$\{ \sipser \}^O$]
    Given an oracle $O$ and $k \in \N$, let $O(i_1, \ldots, i_k)\in \{0,1\}$ be some fixed known Boolean value, for $i_1, \dots, i_k \in \{ 1, \dots, n^{1/(k+1)} \}^k $.
    \(\{\sipser\}^O\) is \(1\) if 
    \[\ \exists i_1 \forall i_2 \exists i_3 \dots Q i_k : O(i_1, i_2, i_3, \dots, i_k) \]
    evaluates to true, and \(0\) otherwise.
\end{definition}

Given the oracle, any machine solving $\{\sipser\}^O$ will have to query $O$ up to $N$ times, while also compounding the answer into a single result according to the quantifiers given in the problem definition.
We now show that $\SPHlevel{k+1}^{O}$ contains this problem, due to the structure of its quantifiers.
In contrast, due to the lack of the $(k+1)$st quantifier, we show $\DRPHlevel{k}^O$ does not contain this problem.

As $\SPHlevel{k+1}$ is equivalent to solving a quantified Boolean formula with $k$ quantifiers, and $\{ \sipser \}^O$ corresponds to verifying \[\exists i_1 \forall i_2 \exists i_3 \dots Q i_k : (i_1, i_2, \dots, i_k) \in O,\] 
it is easy to see that \(\{ \sipser \}^O\in \SPHlevel{k}^{O}\).

\begin{lemma}
\label{thm:LevelSep}
    For a random oracle $O$, with probability $1$ it holds that $\{ \sipser \}^O \not\in \SRPHlevel{k}^O$ and $\{ \sipser \}^O \not\in \PRPHlevel{k}^O$ where $n$ is the input size and $N=2^n$.
\end{lemma}
\begin{proof}
    If $\{ \sipser \}^O \in \SRPHlevel{k}^O$ or $\{ \sipser \}^O \in \PRPHlevel{k}^O$ for a random oracle $O$ with probability $\alpha > 0$, by \Cref{lem:oneTM}, there exists a machine $M^O$ that establishes that $\{ \sipser \}^O \in \SRPHlevel{k}^O$ or $\{ \sipser \}^O \in \PRPHlevel{k}^O$ with probability $\beta > 0.99$.
    Then, \Cref{thm:CircuitLowerBoundsImplyNoMembershipInRPH}
    shows there must exist a Boolean circuit for all $n$ of depth $k+1$, bottom fan-in $\polylog(N)$ and size at most \( N^{\bigO(\polylog(N))} \) that computes $\sipser(\truthtable(O, n))$ for at least some $0.99 \cdot 2^N$ of all possible inputs.
    Note that as $N$ approaches infinity, $N^{-\Omega(1/(k+1))}/N$ approaches $0$.
    So, for any $\beta > 0.99$ and large enough $N$,  \Cref{thm:averageCaseDepthHierarchyTheorem} shows that no circuit of depth $k+1$, size $\quasipoly{N}$ and bottom fan-in $\polylog(N)$ can compute $\sipser$ on more than a $0.5+N^{-\Omega(1/(k+1))}$-fraction of inputs.
    As \( N^{\bigO(\polylog(N))} = \quasipoly{N}\) and the circuit has bottom fan-in $\polylog(N)$ per \Cref{rem:FaninAtmostPoly}, we reach a contradiction and prove that for a random oracle $O$, $\{ \sipser \}^O \not\in \SRPHlevel{k}^O$ and $(\{ \sipser \}^O \not\in \PRPHlevel{k}^O$.
\end{proof}

We showed that $\{ \sipser \}^O \in \SPHlevel{k}^{O}$, $\{ \sipser \}^O \not\in \SRPHlevel{k}^O$, and $\{ \sipser \}^O \not\in \PRPHlevel{k}^O$.
From this, it follows that $\SPHlevel{k}^{O}\not\subseteq\SRPHlevel{k-1}^O$ and $\SPHlevel{k}^{O}\not\subseteq\PRPHlevel{k-1}^O$ for a random oracle $O$, proving \Cref{thm:RPHlayers}.

This result leads us to the following corollary:
\begin{corollary}
    For a random oracle $O$, it holds that $\SRPHlevel{k}^O \subsetneq \SRPHlevel{k+1}^O$.
\end{corollary}
As $\SPHlevel{k+1}^O \subseteq \SRPHlevel{k+1}^O$ and $\SRPHlevel{k}^O \subseteq \SRPHlevel{k+1}^O$, the result immediately shows that $\SRPHlevel{k}^O \subsetneq \SRPHlevel{k+1}^O$.

All proofs follow analogously for $\{ \tilde{F}^{k+1}_N \}^O$ and $\Pi$-layers;
the complement of $\{ \sipser \}^O$ is reached by substituting all $\forall$ quantifiers with $\exists$ quantifiers and vice versa, and negating the input.
As such, we conclude the following corollaries:

\begin{corollary}
        For any $k\geq 0$,
        there exists an oracle $O$, such that  $\PPHlevel{k+1}^{O} \nsubseteq \PRPHlevel{k}^O$. 
\end{corollary}

\begin{corollary}
    There exists an oracle $O$ such that $\PRPHlevel{k}^O \subsetneq \PRPHlevel{k+1}^O$.
\end{corollary}

\subsection{Separating $\RPH$ and $\BQP$}
\label{sec:BQPvsRPH}
Raz and Tal in \cite{RT22} prove the existence of an oracle that separates $\PH$ and $\BQP$. They first present two distributions that are easy to distinguish with a $\BQP$ machine but is hard to distinguish using a classical circuit of quasi-polynomial size and constant depth. Then they use these distributions to construct \emph{distributional oracles} as done in \cite[Section~2.6]{FSUV13} (which in turn is inspired from \cite{Aaronson2010}) to give their oracle separation result. We observe that their result can be straightforwardly adapted to prove the existence of an oracle that separates $\RPH$ and $\BQP$. We state and discuss the results here for completeness. More formally, we are able to prove the following statement.

\BQPnotinRPH*

We will use the following definitions and results from Raz and Tal's \cite{RT22} to prove \Cref{thm:BQPnotinRPH}.
\begin{definition}[Algorithmic advantage]
\label{def:AlgorithmicAdvantage}
Let $\mathcal{D},\mathcal{D'}$ be two probability distributions over a finite set $X$. 
We say an algorithm $\algo{A}$ \textit{distinguishes between $\mathcal{D},\mathcal{D'}$ with advantage $\varepsilon$} if
\begin{equation*}
    \varepsilon=|\underset{x \sim \mathcal{D}}{\probability}[\algo{A} \text{ accepts }x] - \underset{x' \sim \mathcal{D'}}{\probability}[\algo{A} \text{ accepts } x']|.
\end{equation*}
\end{definition}

Note that Raz and Tal used the sets $\{-1,+1\}$ instead of $\{0,1\}$. We adopt their notation here, as it makes it easier to see how we use their results. Additionally, we denote by $\mathcal{U}_\ell$ the uniform distribution on the set $\{-1,+1\}^{\ell}$, i.e., we randomly draw a vector of length $\ell$, where each entry is equally likely to be $-1$ or $+1$. With that, Raz and Tal in \cite{RT22} proved the following.

\begin{theorem}[{\cite[Theorem~1.2]{RT22}}]
\label{thm:MainRazTalTheorem}
There exists an explicit and efficiently samplable distribution $\mathcal{D}_1$ over inputs in $\{\pm 1 \}^{N_1}$ for $N_1$ sufficiently large, such that:
\begin{enumerate}
    \item there exists a quantum algorithm that makes $\polylog(N_1)$ queries to the input, runs in time $\bigO(\polylog(N_1))$ time, and distinguishes between $\mathcal{D}_1$ and $\mathcal{U}_{N_1}$ with probability $1-2^{-\polylog(N_1)}$.
    \item No Boolean circuit of $\quasipoly{N_1}$ size and constant depth distinguishes between $\mathcal{D}_1$ and $\mathcal{U}_{N_1}$ with advantage better than $\polylog(N_1)/\sqrt{N_1}$.
\end{enumerate}
\end{theorem}
The proof of \Cref{thm:MainRazTalTheorem} (and most part of our result) is based on the following claims using the parameters chosen as follows. For all $n \in \N$ and for all $\varepsilon \in [0,1]$, let $N=2^n$, $\delta \geq \frac{1}{2^{\polylog(N)}}$ (we fix $\delta=\frac{1}{n^2}$), $m=32 \cdot \frac{\ln(1/{\delta})}{\varepsilon^2}$ and $N_1=2Nm$. Furthermore, let $\mathcal{D}_1$ be the distribution satisfying conditions stated in \Cref{thm:MainRazTalTheorem}. This distribution $\mathcal{D}_1$ is generated by taking a concatenation of $m$ many independent random variables with
distribution $\mathcal{D}$ which is the distribution over inputs in $\{\pm1\}^{2N}$ constructed in \cite[Section~4]{RT22}. Then,

\begin{claim}[{\cite[Claim~8.1]{RT22}}]
\label{thm:QuantumAlgoQ1}
 There exists a quantum algorithm $Q_1$ making $\bigO (m)$ queries and running in $\bigO( m \cdot \log N)$ time, such that, $\underset{z \sim \mathcal{D}_1}{\probability}[\algo{Q_1} \text{ accepts }z] \geq 1-\delta$ and $\underset{z \sim \mathcal{U}_{N_1}}{\probability}[\algo{Q_1} \text{ accepts }z] \leq \delta$.
\end{claim}
Interpreting the statement of \Cref{thm:QuantumAlgoQ1} in terms of the \emph{algorithmic advantage} tells us that there exists a quantum algorithm $\algo{Q}_1$ that makes $\bigO(m)$ queries and runs in $\bigO( m \cdot \log N)$ time and distinguishes distributions $\mathcal{D}_1$ and $\mathcal{U}_{N_1}$ with advantage at least $1-2\delta$, i.e.,
\begin{equation}
\label{eq:QuantumAlgoDistAdvantage}
\underset{z \sim \mathcal{D}_1}{\probability}[\algo{Q_1} \text{ accepts }z] - \underset{z \sim \mathcal{U}_{N_1}}{\probability}[\algo{Q_1} \text{ accepts }z]\geq 1-2\delta.   
\end{equation}

In the following claim, we use $\underset{x \sim \mathcal{D}}{\mathbb{E}}$ to denote the expected value by taking $x$ according to the distribution $\mathcal{D}$.

\begin{claim}[{\cite[Claim~8.2]{RT22}}] 
\label{thm:SizeVsAdvantage}
Let $C$ be any circuit of size $s$ and depth $d$. Then,
\begin{equation*}
    |\underset{x \sim \mathcal{D}_1}{\mathbb{E}}[C(x)] - \underset{x' \sim \mathcal{U}_{N_1}}{\mathbb{E}}[C(x')]| \leq m \cdot 32 \varepsilon \cdot (c \cdot \log s)^{2(d-1)} \cdot N^{-\frac{1}{2}}.
\end{equation*}
\end{claim}

Using \Cref{thm:SizeVsAdvantage} along with the value of $s=N_1^{\bigO(\polylog(N_1))}$ obtained from \Cref{sub:InterpretationCircuitUpperBounds} we get the following corollary.

\begin{corollary}
\label{thm:RPHcircuitSizeAndNoGoTheorem}
Let $C$ be any circuit of size $s=N_1^{\bigO(\polylog(N_1))}$ and depth $d$. Then,
\begin{align*}
\begin{split}
    |\underset{x \sim \mathcal{D}_1}{E}[C(x)] - \underset{x' \sim \mathcal{U}_{N_1}}{E}[C(x')]| & \leq m \cdot 32 \varepsilon \cdot (c \cdot \log s)^{2(d-1)} \cdot N^{-\frac{1}{2}} \\
   & \leq m \cdot 32 \varepsilon \cdot (c \cdot \polylog(N_1))^{2(d-1)} \cdot N^{-\frac{1}{2}} \\
   & \leq 2m^2 \cdot 32 \varepsilon \cdot (c \cdot \polylog(N_1))^{2(d-1)} \cdot N_{1}^{-\frac{1}{2}}
\end{split}
\end{align*} 
\end{corollary}

Similar to the case with Raz and Tal, we also fix $\varepsilon = \frac{1}{24 \cdot \ln{N}}$ so that any $d$-depth circuit of size $s$ has diminishing advantage of distinguishing distributions $\mathcal{D}_1$ and $\mathcal{U}_{N_1}$ as seen in \Cref{thm:RPHcircuitSizeAndNoGoTheorem}. Rewriting \Cref{thm:RPHcircuitSizeAndNoGoTheorem} in terms of probabilities and substituting the values choosen for $m, \varepsilon$ gives us 
    \begin{equation}
    \label{eq:ClassicalCircDistAdvantage}
        |\probability_{x \sim \mathcal{D}_1} [C(x)=1] - \probability_{x \sim \mathcal{U}_{N_1}} [C(x)=1]| \leq \frac{\polylog(N_1)}{\sqrt{N_1}}.
    \end{equation}

Having fixed all the required parameters, we now proceed to prove \Cref{thm:BQPnotinRPH}.

\begin{proof}[Proof of \Cref{thm:BQPnotinRPH}]
Recall that, for all $n \in \mathbb{N}$ we defined $N=2^n$, $\delta = \frac{1}{n^2}$, $m=32 \cdot \frac{\ln(1/{\delta})}{\varepsilon^2}$, $N_1=2Nm$ and $\varepsilon = \frac{1}{24 \cdot \ln{N}}$. 
\begin{enumerate}
    \item \label{Item:TheLanguage} \textbf{The language $L$}. Let $L$ be a uniformly random language in $\{1\}^*$, i.e., for every $n \in \mathbb{N}$ with probability $\frac{1}{2}$ the string $1^n$ is in $L$.
    
    \item \label{Item:DistOracle} \textbf{The oracle $\oracle{O}$.} For every $n \in \mathbb{N}$, if $1^n \in L$, then $\oracle{O}$ samples $x_n \in \{\pm 1\}^{N_1}$ from the distribution $\mathcal{D}_1$ otherwise $\oracle{O}$ samples $x_n$ from the uniform distribution $\mathcal{U}_{N_1}$. The string $x_n$ can be interpreted as the Boolean function $f_n:\{\pm 1\}^{\lceil \log{N_1} \rceil} \rightarrow \{\pm 1\}$ that describes the oracle $\oracle{O}$ restricted to membership queries of strings of length $\lceil \log{N_1} \rceil$.
    
    \item \label{Item:BQPmachine} \textbf{A $\BQP^{\oracle{O}}$ machine deciding $L$.} Let $M^{\oracle{O}}$ be a $\BQP$ machine with query access to the oracle $\oracle{O}$ from \Cref{Item:DistOracle} above. Given an input $1^n$, the machine $M^\oracle{O}$ runs the algorithm $Q_1$ from \Cref{thm:QuantumAlgoQ1} on the $N_1$-length string provided by $\oracle{O}$ and accepts/rejects according to $Q_1$. Clearly, $M^\oracle{O}$ is a $\BQP$ machine because $Q_1$ runs in $\bigO( m\cdot \log N)=\poly(n)$ time. Furthermore, $M^\oracle{O}$ decides $L$ because of the following reasons.
    \begin{enumerate}
        \item \label{item:acceptingProb} If $1^n \in L$, the oracle $\oracle{O}$ samples $x_n$ from the distribution $\mathcal{D}_1$, and
        \begin{equation*}
            \probability_{x_n \sim \mathcal{D}_1} [Q_1 \text{ accepts } x_n] \geq 1 - \frac{1}{n^2}.
        \end{equation*}
        \item \label{item:rejectingProb} If $1^n \notin L$, the oracle $\oracle{O}$ samples $x_n$ from the distribution $\mathcal{U}_{N_1}$, and
        \begin{equation*}
            \probability_{x_n \sim \mathcal{U}_{N_1}} [Q_1 \text{ accepts } x_n] \leq \frac{1}{n^2}.
        \end{equation*}
    \end{enumerate}
    For any $x \in \{1\}^{*}$, let $L(x) \in \{0,1\}$ denote whether or not $x \in L$; more precisely, $L(1^n) \coloneqq 1$ whenever $1^n \in L$ otherwise $L(1^n) \coloneqq 0$. Consequently we get, for all $n \in \mathbb{N}$ sufficiently large,
    \begin{equation}
    \label{eq:BehaviourOfBQPmachine}
        \probability_{L, \oracle{O}}[M^{\oracle{O}}(1^n) = L(1^n)] \underset{(1)}{=}\frac{1}{2} \cdot \probability_{x_n \sim \mathcal{D}_1} [Q_1 \text{ accepts } x_n] + \frac{1}{2} \cdot \probability_{x_n \sim \mathcal{U}_{N_1}} [Q_1 \text{ rejects } x_n] \underset{(2)}{\geq} 1 - \frac{1}{n^2}.
    \end{equation}    
    The equality (1) stems from the construction of $L$ and $\oracle{O}$ as stated in \Cref{Item:TheLanguage} and \Cref{Item:DistOracle}, respectively and from the behaviour of $M^{\oracle{O}}$ which in this case is the behaviour of Algorithm $Q_1$ on an input given by the oracle $\oracle{O}$. The inequality (2) stems from the behaviour of $Q_1$ as stated in \Cref{item:acceptingProb} and \Cref{item:rejectingProb} (alternatively as stated in \Cref{eq:QuantumAlgoDistAdvantage}).
    
    \item \label{Item:NoPHmachine} \textbf{No $\RPH^{\oracle{O}}$ machine can decide $L$ conditioned on $\BQP^{\oracle{O}}$ deciding $L$.} Let $d$ be a constant and let $C$ be any circuit of size $N_1^{\bigO(\polylog(N_1))}$ and depth $d$. Let $A^\oracle{O}$ denote a $\RPH$ machine with oracle access to $\oracle{O}$. Invoking the connection between constant-depth circuits and $\RPH$ as witnessed in \Cref{thm:CircuitLowerBoundsImplyNoMembershipInRPH} and \Cref{thm:RPHcircuitSizeAndNoGoTheorem} (more specifically \Cref{eq:ClassicalCircDistAdvantage}), along with the construction of language $L$ and oracle $\oracle{O}$ stated in \Cref{Item:TheLanguage} and \Cref{Item:DistOracle} (respectively) gives, for all $n \in \mathbb{N}$ sufficiently large,
    \begin{equation}
    \label{eq:BehaviourOfPHmachine}
    \probability_{L,\oracle{O}} [A^{\oracle{O}}(1^n)=L(1^n)] \leq \frac{1}{2}+\frac{\polylog(N_1)}{2\sqrt{N_1}}.
    \end{equation}
    Let $Q(n)$ denote the event that $M^{\oracle{O}}(1^n)=L(1^n)$ and $E_A(n)$ denote the event that $A^{\oracle{O}}(1^n)=L(1^n)$. Then from \Cref{eq:BehaviourOfBQPmachine,eq:BehaviourOfPHmachine} we have 
     $\probability_{L, \oracle{O}}[Q(n)] \geq 1 -\frac{1}{n^2}$ and $\probability_{L, \oracle{O}}[E_A(n)] \leq \frac{1}{2} + \frac{\polylog(N_1)}{\sqrt{N_1}}$, respectively. Thus,
     \begin{equation*}
         \probability_{L, \oracle{O}} [E_A(n) \mid Q(n)] \leq \frac{\probability_{L, \oracle{O}} [E_A(n)]}{\probability_{L, \oracle{O}} [Q(n)]} \leq \frac{\frac{1}{2} + \frac{\polylog(N_1)}{\sqrt{N_1}}} {1 -\frac{1}{n^2}}< 1.
     \end{equation*}
    Because of the independence of different input lengths, we further get
    \begin{equation*}
         \probability_{L, \oracle{O}} [E_A(n_0) \land E_A(n_0+1) \land \ldots  \mid Q(n_0) \land Q(n_0 +1) \land \ldots] =0;
     \end{equation*}
     here $n_0$ denotes the sufficiently large $n$ such that for all $n\geq n_0$ both \Cref{eq:BehaviourOfBQPmachine,eq:BehaviourOfPHmachine} hold. As the number of $\RPH$ machines is countably infinite, using a union bound over these machines gives
     \begin{equation*}
         \probability_{L, \oracle{O}} [\exists A \text{ s.t. } E_A(n_0) \land E_A(n_0+1) \land \ldots  \mid Q(n_0) \land Q(n_0 +1) \land \ldots] =0.
     \end{equation*}
     So conditioned on $M^{\oracle{O}}(1^n)=L(1^n)$ for all $n \geq n_0$, we get $L \notin \RPH^{\oracle{O}}$ with probability $1$ over all choices of $L, \oracle{O}$. Thus (by hardwiring $L$ on $1^n$ for $n < n_0$ in the $\BQP$ machine), there exists an oracle $O$ for which $\BQP^{\oracle{O}} \nsubseteq \RPH^{\oracle{O}}$.\qedhere
\end{enumerate}
\end{proof}

\bibliographystyle{alpha}
\bibliography{lib,ER,oracle}

\appendix

\section{Foundations}
\label{sec:preliminaries}


This section serves as a self-contained introduction to basic concepts of complexity theory
to the readers that are not very familiar with it.
It also serves as reference to make precise what exactly we mean with our notation.
This can be useful as the terms are sometimes overloaded within the complexity theory literature.
It is also useful to spell out certain details in our proofs more explicitly.

We start to recall the definition of a classical Turing machine.
Thereafter, we explain how we extend it to real Turing machines 
and how this is equivalent to standard real computation models in the literature.

\subsection{Turing Machines}
\label{sec:TuringMachines}
As there are many different definitions of Turing machines, for the sake of being self-contained, 
we provide a definition that is very close to the definition given by Arora and Barak~\cite{AB09}.
Before we start, we note that $S^* = \bigcup_{k\in \N} S^k$, for a given set $S$.

\begin{definition}[Turing Machine {\cite[Section~1.2]{AB09}}]
A Turing machine receives some input $\{0,1\}^*$ does some computations and outputs again a word in $\{0,1\}^*$
or does not halt. In which case, we say the output is unspecified, abbreviated with the symbol $\bot$.
Formally, a Turing Machine $M$ is specified by a tuple $(\Gamma, Q, \delta)$.
Those components are called the alphabet $\Gamma$, the set of states~$Q$, and the transition function $\delta$.
\begin{itemize}[leftmargin=2cm]
     \item[alphabet] 
    The \textit{alphabet} is denoted by $\Gamma = \{0, 1, \square, \rhd\}$, which comprises the symbols $0$ and $1$, the designated "blank" symbol $\square$, and the "start" symbol $\rhd$. 
    \item[tape]
    The Turing machine $M$ operates on a single tape, which consists of a countably infinite sequence of cells indexed by the natural numbers. A \textit{tape head} moves across the tape, indicating the currently active cell. At each step, the tape head can move one cell to the left or right, or stays at the current location.

    \item[states]
    The set $Q = \{1,\ldots,q\}$ denotes all possible states of the Turing machine $M$. 
    We assume that $Q$ includes a designated starting state $q_{\texttt{start}} = 1$ and a designated halting state $q_{\texttt{halt}} = q$.
    \item[transition] A function $\delta: Q \times \Gamma \rightarrow Q \times \Gamma \times \{ \leftarrow, \circ, \rightarrow \}$.
    This function is called the \textit{transition function} and describes how $M$ changes each step.
    The input is the current state and the symbol on the tape at the head position.
    The output is the new state, a new symbol and a movement specification for the head.
    \item[initialization] 
    The tape is initialized with the start symbol $\rhd$ in its first cell, followed by a sequence of non-blank symbols representing the \textit{input} in binary. All other cells are initialized to the blank symbol $\square$. The tape head begins at the first cell on the left end of the tape. The machine starts in the special state $q_{\texttt{start}}=1$.

    \item[computation]
    In each step of the computation, the function $\delta$ is applied.
    The state of $M$ and the symbol read by the tape head will be the input for $\delta$;
    after the computation, $\delta$ specifies the new state, writes a new symbol to the tape at the current head position and moves the tape head either left or right by one position.
    If the tape head is already in the leftmost 
    position, the tape head stays put.
    (Recall that each position of the tape corresponds to a natural number starting with $0$ at the very left and going to infinity to the right.)
    The computation repeatedly applies the $\delta$-function until it reaches the special state $q_\texttt{halt}$.
    Then the computation halts.
    \item[output] When the Turing machine halts, i.e., enters state $q_{\texttt{halt}} = q$, it outputs everything to the left of the tape head, ignoring the blank symbols and the start symbol.
    Note that the output is always a finite sequence of binary symbols.
\end{itemize}

\end{definition}
In our definition of binary Turing machine, we have made some choices.
\begin{itemize}
    \item For example, it is possible to allow multiple tapes.
However, it is known that multiple tapes can be simulated on a single tape.
    \item It is also common to restrict to a binary alphabet without the start and the blank symbol.
The idea is that we map the four symbols $\{0, 1, \square, \rhd\}$ to the pairs $\{00, 01, 10, 11\}$.
However, this comes at a cost of a more complicated input and output encodings, which we wanted to avoid.
On the other hand, if the alphabet $\Gamma$ is chosen to be larger, then it is easy to simulate the larger alphabet with our alphabet.
    \item In a similar spirit, it is possible to allow for multiple heads. Again, this is known not to change the computational power.
    \item Also note that, technically, $\Gamma$ is fixed for us and thus does not need to be specified, we do this to emphasize that different alphabets are possible.
    \item Note also that what we did is equally valid to have the tape bidirectionally infinite.
    We see neither advantages nor disadvantages to either choice, and thus picked the arguably simpler choice to be infinite in only one of the two directions.
    \item Sometimes Turing machines have an accepting or rejection state.
    This can be simulated with making the output being $1$ or $0$.
    \item The output can also be defined as all non-blank symbols on the tape.
    Again, this is equivalent to our choice, see for example~\cite{vEB12}.
\end{itemize}

We will also call a Turing machine a \textit{binary Turing machine}, to contrast it to the \textit{real Turing machine} defined below.
Intuitively, a real Turing machine extends the definition of a binary Turing machine by 
adding one extra tape which can hold real values.
Furthermore, we allow it to calculate basic arithmetic operations in one step: $\{ +, -, \cdot, \div \}$ on the real tape.
We repeat also the parts of the real Turing machine that are identical to the binary Turing machine in order to have a self-contained definition.

\begin{definition}[Real Turing machine]
\label{def:RealTuringMachine}
    A real Turing machine receives some input $\{0,1\}^*\times \R^*$, does some computations and 
    outputs again a word in $\R^*$ or does not halt. 
    In the latter case, we say the output is unspecified, abbreviated with the symbol $\bot$.
    We specify a real Turing machine by the tuple $(\Gamma_\bin, \Gamma_\real, Q, \delta)$.
    Those components are called the alphabets $\Gamma_\bin,\Gamma_\real$, the set of states~$Q$, and the transition function $\delta$.
    \begin{itemize}[leftmargin=2cm]
    \item[alphabets] We specify the \textit{binary alphabet} by  $\Gamma_\bin = \{0, 1, \square, \rhd \}$.
    The \textit{real alphabet} $\Gamma_\real = \R \cup \{ \square, \rhd\}$ has real numbers instead of binary numbers.
    \item[tapes] We have a binary and a real tape.
    Each tape contains a cell for each natural number.
    We think of the $0$ position as the leftmost position ``going to the right'' to infinity.
    Each cell of the binary tape contains a symbol from the binary alphabet and likewise each cell from the real tape contains a real symbol.
    The binary tape has one head and the real tape has three heads.
    We denote the real head positions with $h,i,j$, where $h$ is a write head and $i,j$ are reading heads.

     \item[states]
    The set $Q = \{1,\ldots,q\}$ denotes all possible states of the Turing machine $M$. 
    We assume that $Q$ includes a designated starting state $q_{\texttt{start}} = 1$ and a designated halting state $q_{\texttt{halt}} = q$.
    \item[operations] The Turing machine, can only write on the cell $x_h$ of the head $h$.
    It can use the content of the cells $x_i$ and $x_j$ of the heads $i,j$.
    The set of operations $\operations$ is the following:
    \begin{enumerate}
        \item \textit{Idle}: $x_h = x_h$.
        \item \textit{Addition}: $x_h = x_i + x_j$.
        \item \textit{Subtraction}: $x_h = x_i - x_j$.
        \item \textit{Multiplication}: $x_h = x_i \cdot x_j$.
        \item \textit{Division}: $x_h = \frac{x_i}{x_j}$ 
        \item \textit{Write constant}: $x_h \leftarrow x$ where $x \in \{0,1,\square, \rhd\}$.
        \item \textit{Copy}: $x_h \leftarrow x_i$. 
    \end{enumerate}
    Note that there are $10$ different possible operations.
    \item[transition] 
    The \textit{transition function} $\delta$ takes as input the current state, the content 
    of the active cell on the binary tape, and the sign
    ($\signset = \{+,0,-\}$) of the reading head $i$ of the real tape.
    (By convention, we assume that the sign of the start and the blank symbol is $0$.)
    It returns a new state, a new binary symbol, an operation and a specification of the head movements $\{\leftarrow, \circ, \rightarrow\}$.
    Formally, we write \[\delta : Q \times \Gamma_\bin \times \signset \rightarrow Q \times \Gamma_\bin \times \operations \times \{\leftarrow, \circ, \rightarrow\}^4.\]
     \item[initialization] 
    The binary tape is initialized with the start symbol $\rhd$ in its first cell, followed by a finite sequence of $\{0,1\}$ symbols representing the \textit{input} in binary. 
    All other cells are initialized to the blank symbol $\square$. 
    The tape head begins at the first cell on the left end of the tape. 
    The real tape is initialized equally with the real input instead of the binary input.
    The machine starts in the special state $q_{\texttt{start}}= 1$.
    \item[computation]
    In each step of the computation, the function $\delta$ is applied.
    The symbol on the binary tape, the current state and the sign of the first real reading head is processed.
    The Turing machine goes into a new state, writes on the binary tape directly and performs the specified operation on the real tape.
    All heads are moved according to $\delta$ as well.
    If any of the tape heads are already in the leftmost 
    position, it cannot move further to the left and stays put instead.
    The computation repeatedly applies the $\delta$-function until it reaches the special state $q_\texttt{halt}$.
    Then the computation halts.
    \item[output]
    When the Turing machine halts, i.e., reaches state $q$, it outputs everything to the left of the real tape head $h$, ignoring the blank symbols and the start symbol.
    Note that the output can be interpreted in a binary way, as $\{0,1\}^* \subset \R^*$.
    \end{itemize}
\end{definition}

Again there are choices to be made in the way that we define a real Turing machine.
Clearly, the BSS model of computation serves as an inspiration~\cite{BCSS98}.
Different sources also differ in part in the way that BSS machines are described~\cite{BSS89}.
For example, the original paper by Blum Shub and Smale had two heads~\cite{BSS89} and a one directional tape.
In contrast, the model in the celebrated book by Blum, Cucker, Shub and Smale~\cite{BCSS98} had no head at all.
Instead, they had the possibility to ``shift'' the entire input one step to the right or one step to the left.
We found the description by Gr\"{a}del and Meer~\cite{GM95}, particularly pleasing to read.

We want to highlight some of the choices we made that differ to other definitions in the literature:
\begin{itemize}
    \item BSS machines often are allowed to have access to arbitrary real constants. 
    As we pointed out earlier this leads to a lot of power.
    We therefore deliberately do not allow our real Turing machine to use any constants apart from zero and one.
    \item We decided not to allow arbitrary rational functions.
        Instead, we allow standard arithmetic operations that
        can simulate those rational functions, so the model would remain as simple as possible.
    \item BSS machines typically allow branching on the sign of some fixed polynomials $\beta$ 
    evaluated on a constant number of real cells.
    The polynomial $\beta$ is part of the BSS-machine.
    We can easily simulate this branching, by computing $\beta$ using the basic arithmetic operations and then check the sign 
    of the cell containing it.
    For example, if we want to simulate the instruction $x_i \leq x_j$, we compute $x_h = x_i - x_j$,
    and then check the sign of $x_h$.
    \item Also, the transition function has no direct access to the numbers in the real tape, but only to the sign
        of one of the heads. 
        In this way, we do not have to restrict the set of possible transition functions.
    \item The first definition of a BSS-machines stems from 1989~\cite{BSS89}.
        Yet, there are several different authors who gave machine models 
        that were capable to work with more general structures than a finite alphabet.
        
        We find the description of Gr\"{a}del and Meer~\cite{GM95} very accessible.
        The original definition allowed two heads that can be moved.
        But secretly, the branching and computation instructions allowed access
        to more registers.
        Later, in the book~\cite{BCSS98} from 1998 the authors define a machine model that has no head at all.
        Yet, it has a shift operation that moves the entire tape one step to the left or one step to the right.
        Which is again a head in disguise.
        We are not sure why the authors, try to avoid using heads that operate on the tape.
        We take the opposite approach and are generous with the number of heads we allow,
        in order to have a simpler and more intuitive description.
    \item 
        In the literature, the word random access is used in different ways.
        In the 1980s, random access referred to the fact that any part of the memory could access.
        In comparison, we can only access the top element of a stack.
        Yet nowadays, if people talk about the word RAM, it stands for random access memory,
        which means that we can use the content of a register as an address to another register.
        We can manipulate the contents of this register, which allows efficient access to otherwise distant registers. 
        It is that ability, that allows us to have a fine-grained complexity analysis,
        which is not so useful on a Turing machine, as most of the running time is spent by moving the heads forth and back.
    \item We also want to point out that division is not strictly needed due to the fact that, for positive $a,b,c,d$ it holds that
        \[\frac{a}{b} < \frac{c}{d} \Leftrightarrow cb - ad > 0.\]
        Therefore it is possible to store nominator and denominator of each number separately and simulate all other operations, including determining the sign of a number with a constant number of extra calculations. 
    \item BSS machines typically define functions $f : \R^* \rightarrow \R^*$. 
        And thereafter, we may restrict out inputs to be binary~\cite{BC09}. 
        We explicitly have a binary tape to emphasize that 
        the model is meant to take binary inputs.
\end{itemize}

To use Turing machines in a sensible way, we do need to be able 
to have a notion of running time of a Turing machine.

\begin{definition}[Runtime]
    Recall that a real Turing machine $M$ defines a function $f : \{0,1\}^*\times \R^* \rightarrow \{0,1\}^*\times \R^* \cup \bot$.
    The size of the input $(x,y) \in \{0,1\}^*\times \R^*$ is the number of binary and real symbols.
    The runtime of a Turing machine on input $(x,y)$ is the number of computation steps until it halts.
    We say $M$ has running time $T : \N \rightarrow \N$, if 
    \begin{itemize}
        \item it halts on any input, and
        \item for any input of size $k$ it halts in at most $T(k)$ steps.
    \end{itemize}
\end{definition}

\begin{lemma}
\label{lem:NrRTMsinstates}
    There are at most 
    $(3240q)^{12 q}$ real Turing machines with $q$ states.
\end{lemma}
\begin{proof}
    Note that if $\delta$ is specified, essentially the rest of the real Turing machine is specified as well.
    \Cref{tab:setsizes} gives the size of each set involved in the transition function.
    \begin{table}
        \centering
        \begin{tabular}{|cc|}
        \hline
        \textbf{Set} & \textbf{Size} \\
        \hline
        $Q$ & $q$ \\
        $\Gamma_{\text{bin}}$ & $4$ \\
        $\signset$ & $3$  \\
        $\operations$ & $10$\\
        $\{\leftarrow, \circ, \rightarrow\}$ & $3$ \\
        \hline
        \end{tabular}
        \caption{Sizes of different sets involved in the transition function.}
        \label{tab:setsizes}
    \end{table}
Thus, the number of possible inputs equals $a= q\cdot 4 \cdot 3 = 12 q$
and the number of possible outputs equals $b = q \cdot 4 \cdot 10 \cdot 3^4 = 3240 q$.
The number of functions from a set of size $a$ to a set of size $b$ equals $b^a = (3240q)^{12q}$.
\end{proof}

\subsection{Real Turing Machine Justification}
The next step is to justify our model of computation.
For that we point out that any function $f : \{0,1\}^* \rightarrow \{0,1\}^*$ computable in polynomial time on a binary Turing machine is trivially also computable in polynomial time on a real Turing machine, as we can just ignore the real tape.
Furthermore, we should make sure that we are not more or less powerful than other standard models of computation.
As comparison, we take the realRAM as defined by Erickson, Hoog and Miltzow~\cite{EvdHM20}.

\begin{lemma}
    \label{lem:realRAMEquivalence}
    Let $f : \{0,1\}^*\times \R^* \rightarrow \{0,1\}^*\times \R^*$ be a function.
    Then $f$ can be computed in polynomial time on a real Turing machine if and only if it can be computed in polynomial time on a realRAM.
\end{lemma}
\begin{proof}
    We first argue that a realRam can simulate a real Turing machine.
    The binary input of the real Turing machine is stored in the word registers (leaving the first 4 word registers empty) and
    the real input is stored in the real registers.
    We can store the position of the four heads in the four separate word registers at the beginning.
    The realRAM can then simulate each computation step easily.
    It can store the current state, update the head positions, perform the correct arithmetic operations and go to the next state.

    We argue now that real Turing machines can simulate the realRAM.    
    We use the known fact that the wordRAM can be simulated by a binary Turing machine.
    Specifically, all standard operations, like indirect accessing, addition, subtraction etc. from the wordRAM can
    be simulated on a binary Turing machine.
    Thus, we only need to show that also the operations on the real registers can be simulated as well.
    Those include casting integers to reals, and real arithmetic operations.
    All standard real arithmetic operations  are part of the definition of the real Turing machine.
    In order to describe casting, let $x$ be the content of an integer register.
    Let us say that $x = b_1\ldots b_k$ in binary, with each $b_i \in \{0,1\}$.
    We see that $x = ( \ldots (2(2(b_1 ) + b_2)+ \ldots ) + b_k$.
    This means we can construct $x = x_k$ by computing $x_1 = b_1$, and 
    $X_{i+1} = 2\cdot x_i + b_{i+1}$. 
    This takes at most a polynomial number of steps.
\end{proof}

\subsection{Oracle Machines}
\label{sec:OracleDefinition}

Sometimes our Turing machines can have access to an oracle, 
then we call it an \emph{oracle} (real) Turing machine. 
The idea is that an \textit{oracle} can be an arbitrary set $O\subseteq \{0,1\}^*$ 
or $O\subseteq \R^*$.
We denote those as \textit{binary} or \textit{real} oracles, respectively.
Unless specified otherwise, our oracles are binary.
The idea of an oracle Turing machine is that it can check if a given $x$ is in $O$ or not in one time step.
Note that there are no restrictions on the set $O$.
So $O$ could also contain the answer to the halting problem.
Given a real Turing machine $M$ and an oracle $O$, we write $M^O$ or $(M,O)$ to refer to $M$ with oracle access to $O$.

\begin{definition}[Oracle (real) Turing machines]
\label{def:OracleRTM}
    We define oracle machines as extensions of normal machines.
    We add one more tape, which we refer to as the \textit{oracle tape}.
    Depending on the oracle, this tape is a binary or real tape and the alphabet is either
    $\Gamma_\bin$ or $\Gamma_\real$. 
    This tape has only one head, and it is only possible to copy to this tape from a different tape or to write a blank or start symbol.
    We interpret the string of symbols ignoring the blank and start symbol as the query $x$.
    The transition function has access to the answer \texttt{y/n} of the query $x\in O$.
    Thus, formally, we extend the operations $\operations$ to include the operations
    \begin{enumerate}
        \item copy on oracle tape $x_o \leftarrow x_h$
        \item write blank or start symbol $x_o \leftarrow \square$ or $x_o \leftarrow \rhd$.
    \end{enumerate}
    Note that the new set of operations has size $13$, as we added 3 possible operations.
    
    Formally, the transition function for a binary oracle function is described by
    \[\delta : Q \times \Gamma_\bin \times \signset \times \{\texttt{y},\texttt{n}\} \rightarrow Q \times \Gamma_\bin \times \operations \times \{\leftarrow, \circ, \rightarrow\}^5.\]
    Similarly, $\delta$ can be adopted for a real oracle functions.
    
\end{definition}

\subsection{Complexity Classes}
\label{sec:ComplexityClasses}

We are now ready to define decision problems and languages.
We think of a decision problem as having some input and a \texttt{yes/no} question.
For example, the input could be a graph, and we want to know if it contains a Hamilton cycle.
To be formal and in a coherent framework, we need to understand two steps.
First, any input can be encoded using $0$'s and $1$'s.
Furthermore, if we only look at the set of inputs with answer \texttt{yes} then we 
get a set $L \subseteq \{0,1\}^*$.
Now, if a machine can solve if some $x$ is in $L$ then it can also solve the decision problem.
We distinguish between binary language ($\subseteq \{0,1\}^*$) and real languages ($\subseteq \R^*$).
However, most of this article deals with binary languages.

\begin{definition}[Function and language recognition]
\label{def:FunctionNlanguageRecognition}
Let $f:\{0,1\}^{*} \rightarrow \{0,1\}$ be a function with Boolean output. 
The language corresponding to $f$, is the set of strings $L \subseteq \{0,1\}^{*}$ on which $f$ evaluates to $1$, i.e., $\forall x \in \{0,1\}^{*}$,
\begin{equation*}
    x \in L \iff f(x)=1.
\end{equation*}

We say a real Turing machine $M$ \emph{decides} a language $L \subseteq \{0,1\}^{*}$ if for all $x \in \{0,1\}^*$, whenever $M$ is initialized to the start configuration on input $x$ (i.e., goes to the $q_\texttt{start}$ state), then it halts (i.e., goes to the $q_\texttt{halt}$ state) with $f(x)$ written on its output tape.
\end{definition}

\paragraph{Complexity Classes.}
We are now ready to define \textit{complexity classes}, simply as a set of languages.
In other words, $C \subseteq \power(\{0,1\}^*)$.
Given a machine model $M$, we say $C$ is \textit{defined by} $M$,
if for every language $L$ it holds that
\[L(\prop, O) \in C \Leftrightarrow M \text{ decides } L.\]
We write $C^O$ if $C$ is defined by machine model $M$ and
machine model $M$ with oracle access to $O$ decides~$C^O$.

Typically, if a complexity class $C$ comes with two natural machine models $M_1$ and $M_2$ then adding an oracle $O$ to either machine leads to the same complexity class. For example, \PSPACE could be defined using polynomial space single tape Turing machines, or using wordRAM and polynomial space.
Adding an oracle $O$ to either of them in a natural way leads to the same complexity class $\PSPACE^{O}$.
In such a case, we can neglect the machine model $M$ in the notation. On the other hand,  $\textsc{IP} = \PSPACE$, but there exists an oracle $O$ such that $\textsc{IP}^O \neq \PSPACE^O$ and that happens because the underlying machine models can access the oracle differently.

We are now ready to define standard complexity classes.

\begin{itemize}
    \item \P is the set of languages that can be decided with a binary Turing machine in polynomial time.
    \item \PSPACE is the set of languages that can be decided with a binary Turing machine in polynomial space.
    \item \NP is defined as follows. 
    $L\in \NP$ if and only if there is a binary Turing machine $M$, with polynomial running time and a polynomial $q$ such that
        \[x \in L \Leftrightarrow \exists w\in \{0,1\}^{q(|x|)} : M(x,w) = 1.\]    
    \item \ER is defined as follows. 
    $L\in \ER$ if and only if there is a real Turing machine $M$, with polynomial running time and a polynomial $q$ such that
        \[x \in L \Leftrightarrow \exists w\in \R^{q(|x|)} : M(x,w) = 1.\]
    \item \BQP is the set of languages that can be decided by quantum computers using polynomial time and with error probability $\leq$ 1/3 on every input. We refer the readers to \cite[Definition~10.9]{AB09} for a detailed definition.
\end{itemize}

\begin{definition}[Three types of polynomial hierarchies, Definition~5.3 in \cite{AB09}]
\label{def:PH}
A language  $L$ is in $\PH$ if there exists a polynomial-time Turing machine $M$, polynomial $q:\N \rightarrow \N$, and a constant $k \in \N$ such that on every $x \in \{0,1\}^{*}$,
\begin{equation*}
    x \in L \iff Q_1 u_1 \in \{0,1\}^{q(|x|)} Q_2 u_2 \in \{0,1\}^{q(|x|)} \cdots Q_k u_k \in \{0,1\}^{q(|x|)} M(x, u_1,\ldots,u_k)=1
\end{equation*}
where $Q_i \in \exists$ or $\forall$.
The $k$th level is denoted by $\DPHlevel{k}$ if at most $k$ quantifier alternations are allowed.
We distinguish further between $\SPHlevel{k}$ and $\PPHlevel{k}$, to indicate that the first quantifier is either existential or universal.

We define the \textit{truly real polynomial hierarchy} \TrulyRPH in the same way, but with every occurrence of $\{0,1\}$
replaced by $\R$ and $M$ a real Turing machine.
That is, the language is real, and the quantified inputs are real.
The levels are denoted by \TrulyDRPHlevel{k}, \TrulySRPHlevel{k}, and \TrulyPRPHlevel{k}.

In this work, we deal with the \textit{real polynomial hierarchy} \RPH, which is defined in the same way as the truly polynomial hierarchy, except that we restrict ourselves to binary languages.
The levels are denoted by \DRPHlevel{k}, \SRPHlevel{k} and \PRPHlevel{k} respectively.
Note that $\SRPHlevel{1}$ is also denoted as $\ER$, $\PRPHlevel{1}$ as $\forall \R$,
$\SRPHlevel{2}$, as $\exists \forall \R$ and so on~\cite{BC09,JKM23, DKMR18, JJ23}.
\end{definition}

\subsection{Boolean Circuits}

This section is dedicated to formally define Boolean Circuits.
We closely follow the definitions by Arora and Barak~\cite{AB09}.

\begin{definition}[Boolean Circuits, Definition~6.1 in \cite{AB09}]
\label{def:BooleanCircuits}
For every $n \in \N$, an \emph{$n$-input, single-output Boolean circuit} $C$ is a directed acyclic graph with the following characteristics:

\begin{enumerate}
    \item \textbf{Inputs:} There are $n$ source vertices, called \emph{inputs}. These vertices have no incoming edges, but may have any number of outgoing edges. Each input vertex is labeled with an index $i \in [n]$, representing the input variables.
    \item \textbf{Output:} There is a single sink vertex, called the \emph{output}. This vertex has no outgoing edges.
    \item \textbf{Gates:} All non-input vertices are called \emph{gates}. Each gate is labeled with one of the logical operations $\land$ (AND), $\lor$ (OR), or $\neg$ (NOT).
    \begin{itemize}
        \item Gates labeled with $\land$ or $\lor$ have an arbitrary number of incoming edges (arbitrary fan-in).
        \item Gates labeled with $\neg$ have exactly one incoming edge (fan-in of 1).
    \end{itemize}
\end{enumerate}

We denote the label of a vertex $v$ by $\operatorname{label}(v)$. 
The \emph{size} of a Boolean circuit $C$, denoted by $|C|$, is the total number of vertices in the circuit.
The \emph{depth} of a circuit is the length of the longest path from a source vertex to the sink vertex.
(The length of a path is measured by the number of edges on the path.)

Given a Boolean circuit $C$ and an input $x \in \{0, 1\}^n$, the output $C(x)$ is computed by assigning a value $\operatorname{val}(v) \in \{0,1\}$ to each vertex $v$ in $C$ as follows:

\begin{enumerate}
    \item If $v$ is an input vertex labeled with $i$, then $\operatorname{val}(v) = x_i$.
    \item If $v$ is a gate, then $\operatorname{val}(v)$ is the result of applying $\operatorname{label}(v)$ to the values of its predecessor vertices (the vertices with edges leading into $v$).
\end{enumerate}

The output $C(x)$ is the value assigned to the output vertex after evaluating the circuit on $x$.
\end{definition}

\begin{remark}
Sometimes in literature, the term Boolean circuit is defined as in \Cref{def:BooleanCircuits} with constant fan-in of at most $2$.
However, we adhere to the other standard in literature, by considering circuits with unbounded fan-in~\cite{AB09}.
\end{remark}

Similar to computation with Turing machines, we can also recognize languages with circuits.

\begin{definition}[Circuit family and language recognition, Definition~6.2 in \cite{AB09}]
\label{def:CircuitFamily}
Let $T: \N \rightarrow \N$ be a function. A \emph{$T(n)$-size circuit family} is a sequence $(C_n)_{n \in \N}$ of Boolean circuits, where for every $n$, $C_n$ has $n$ inputs and a single output and its size $|C_n| \leq T(n)$.

We say a language $L(\prop, O) \in \Size{T(n)}$ if there exists a $T(n)$-size circuit family $(C_n)_{n \in \N}$ such that for every $x \in \{0,1\}^n, x \in L \iff C_n(x)=1$.
\end{definition}

\begin{definition}[$\AC$, Definition~6.25 in \cite{AB09}]
\label{def:AC0circuits}
A language $L$ is in $\AC$ if $L$ can be decided by a family of circuits $(C_n)_{n \in \mathbb{N}}$ where $C_n$ has $\poly(n)$ size and constant depth and is allowed to have gates with unbounded fan-in (i.e., the $\lor$ and $\land$ gates can be applied to more than two bits).
\end{definition}

\subsection{Two Definitions of the Real Polynomial Hierarchy}
\label{sub:TwoRPH-Definitions}
We defined the real polynomial hierarchy using the real Turing machine.
However, it is also possible to define it using first order logical formulas, oracle machines, or Boolean circuits~\cite{M24, ERcompendium}.
We leave open how the latter two definitions would look like.
We recite the definition via first order logical formulas as, to the best of our knowledge, were first described by B\"urgisser and Cucker~\cite{BC09}.

\begin{definition}[\LogicRPH]
\label{def:LogicRPH}
A language $L$ is in $\LogicRPH$ if there exists a binary polynomial-time Turing machine $N$ that returns a quantifier-free formula $\Phi$ with free variables \(u_1, \dotsc, u_k\) on every $x \in \{0,1\}^{*}$ with $|x| = n$.
Furthermore, there must be a constant $k$ and a polynomial $q:\N \rightarrow \N$
\begin{equation*}
    x \in L \iff \exists u_1 \in \R^{q(|x|)} \forall u_2 \R^{q(|x|)} \cdots Q_k u_k \in \R^{q(|x|)} \Phi^{(n)}(u_1,\ldots,u_k)
\end{equation*}
where $Q_i \in \{\exists,\forall\}$.
\end{definition}

It was observed by B\"{u}rgisser and Cucker that \LogicRPH and \RPH define the same set of languages~\cite{BC09}.
The idea of the proof is a Cook-Levin style reduction.
Clearly, a real Turing Machine $M$ can simulate the Turing machine $N$ and evaluate the formula $\Phi_x$ on the variable assignment $u_1,u_2,\ldots , u_k$.
Conversely, all the computations that $M$ will do can be encoded into a formula. 
To do this, there exists a variable $y_{i,t}$ for each cell $i$ of the tape and each time step $t$. 
The content of the cells is determined by 
the quantified variables.
the formula $\Phi_x$ makes sure that all cells are assigned values that could represent consistently a run of the Turing machine $M$. Let us remark that oracle access could be modeled in \LogicRPH by including a relational symbol \(O\) into the set of logical atoms whose semantics would be the output of the oracle \(O\) for the input of the arguments of the relation.


\subsection{Formula Formats}
Finally, we give a rudimentary definition of $\CNF$ and $\DNF$ formulas.
\begin{definition}[$\CNF$ and $\DNF$ Formulas]
\label{def:CNFnDNFformulas}
A $\CNF$ ($\DNF$) formula is a logical Boolean formula in conjunctive (disjunctive) normal form, that is, the formula is a conjunction (disjunction) of one or more clauses where each clause is a disjunction (conjunction) of one or more literals. Otherwise put, it is an $\textsc{AND}$ ($\textsc{OR}$) of $\textsc{OR}$s (\textsc{AND}s).
\end{definition}

\end{document}